\documentclass[runningheads]{llncs}
\usepackage{graphicx}
\usepackage{booktabs}   
\usepackage{float}
\usepackage{xspace}
\usepackage{wrapfig}
\usepackage{subcaption}
\usepackage[shortlabels]{enumitem}

\usepackage{cmap}

\usepackage{shuffle}
\usepackage{syntax}

\usepackage[square,sort,comma,numbers]{natbib}
\usepackage{url}
\usepackage{amsmath} 
\usepackage{amssymb} 
\usepackage{amsfonts}

\usepackage{etoolbox}
\usepackage{mathpartir}
\usepackage{mathtools}
\usepackage{refcount}
\usepackage{stmaryrd}
\usepackage{enumitem}
\usepackage{listings}
\usepackage{caption}
\usepackage{flushend}

\usepackage{array}
\newcolumntype{L}[1]{>{\raggedright\let\newline\\\arraybackslash\hspace{0pt}}m{#1}}
\newcolumntype{C}[1]{>{\centering\let\newline\\\arraybackslash\hspace{0pt}}m{#1}}
\newcolumntype{R}[1]{>{\raggedleft\let\newline\\\arraybackslash\hspace{0pt}}m{#1}}
\newcommand{\tuple}[1]{\left\langle#1\right\rangle}

\usepackage{scalerel}
\usepackage{cancel} 

\usepackage{pifont}
\newcommand{\cmark}{\textcolor{green}{\ding{51}}}%
\newcommand{\xmark}{\textcolor{red}{\ding{55}}}%

\usepackage{pgf}
\usepackage{tikz}
\usetikzlibrary{arrows,automata}
\usetikzlibrary{arrows.meta}
\usepackage{scalefnt}

\usepackage{wrapfig}
\usepackage{lipsum}  

\usepackage{algpseudocode}
\makeatletter
\renewcommand{\ALG@beginalgorithmic}{\footnotesize}
\makeatother
\usepackage{eqnarray}
\usepackage{alltt}
\usepackage{textcomp}
\usepackage{wrapfig,lipsum}
\usepackage{mathsprograms}

\usepackage{listings}
\lstdefinelanguage{JavaScript}{%
  keywords = { async, await, break, case, catch, class, const, continue, debugger, default, delete, do, each, else, export, finally, for, function, if, import, in, instanceof, let, new, of, return, switch, this, throw, try, typeof, var, void, while, with, yield },
  morecomment = [l]{//},
  morecomment = [s]{/*}{*/},
  morestring  = [b]',
  morestring  = [b]",
  sensitive   = true,
}

\lstdefinelanguage{Java10}{
  language      = Java,
  morekeywords  ={ var },
}

\begin{document}

\setlength{\abovedisplayskip}{0.1pt}
\setlength{\belowdisplayskip}{3pt}

\title{Checking Robustness Against Snapshot Isolation\thanks{This work is supported in part by the European Research Council (ERC) under the Horizon 2020 research and innovation programme (grant agreement No 678177).}}
%
%
\author{Sidi Mohamed Beillahi \and
Ahmed Bouajjani \and
Constantin Enea}
\authorrunning{S.M. Beillahi, A. Bouajjani, and C. Enea.}
%
\institute{Universit\'e de Paris, IRIF, CNRS, Paris, France,
\email{\{beillahi,abou,cenea\}@irif.fr}}
\maketitle              
%
\begin{abstract}
Transactional access to databases is an important abstraction allowing programmers to consider blocks of actions (transactions) as executing in isolation. The strongest consistency model is {\em serializability}, which ensures the atomicity abstraction of transactions executing over a sequentially consistent memory. Since ensuring serializability carries a significant penalty on availability, modern databases provide weaker consistency models, one of the most prominent being \emph{snapshot isolation}. In general, the correctness of a program relying on serializable transactions may be broken when using weaker models. However, certain programs may also be insensitive to consistency relaxations, i.e., all their properties holding under serializability are preserved even when they are executed over a weak consistent database and without additional synchronization.

In this paper, we address the issue of verifying if a given program is {\em robust against snapshot isolation}, i.e., all its behaviors are serializable even if it is executed over a database ensuring snapshot isolation. We show that this verification problem is polynomial time reducible to a state reachability problem in transactional programs over a sequentially consistent shared memory. This reduction opens the door to the reuse of the classic verification technology for reasoning about weakly-consistent programs. In particular, we show that it can be used to derive a proof technique based on Lipton's reduction theory that allows to prove programs robust.
\end{abstract}
%
%
%




\everymath{\displaystyle}
\section{Introduction}
Transactions simplify concurrent programming by enabling computations on shared data that are isolated from other concurrent computations and resilient to failures. Modern databases provide transactions in various forms corresponding to different tradeoffs between consistency and availability. The strongest consistency level is achieved with \emph{serializable} transactions~\cite{DBLP:journals/jacm/Papadimitriou79b} whose outcome in concurrent executions is the same as if the transactions were executed atomically in some order. Since serializability carries a significant penalty on availability, modern databases often provide weaker consistency models, one of the most prominent being {\em snapshot isolation} (SI)~\cite{DBLP:conf/sigmod/BerensonBGMOO95}.
Then, an important issue is to ensure that the level of consistency needed by a given program coincides with the one that is guaranteed by its infrastructure, i.e., the database it uses. One way to tackle this issue is to investigate the problem of checking {\em robustness} of programs against consistency relaxations: Given a program $P$ and two consistency models $S$ and $W$ such that $S$ is stronger than $W$, we say that $P$ is robust for $S$ against $W$ if for every two implementations $I_S$ and $I_W$ of $S$ and $W$ respectively, the set of computations of $P$ when running with $I_S$ is the same as its set of computations when running with $I_W$. This means that $P$ is not sensitive to the consistency relaxation from $S$ to $W$, and therefore it is possible to reason about the behaviors of $P$ assuming that it is running over $S$, and no additional synchronization is required when $P$ runs over the weak model $W$ such that it maintains all its properties satisfied with $S$.

In this paper, we address the problem of verifying robustness of transactional programs
for serializability, against {\em snapshot isolation}. Under snapshot isolation, any transaction $t$ reads values from a snapshot of the database taken at its start and $t$ can commit only if no other committed transaction has written to a location that $t$ wrote to, since $t$ started.
%
%
Robustness is a form of program equivalence between two versions of the same program, obtained using two semantics, one more permissive than the other. It ensures that this permissiveness has no effect on the program under consideration. The difficulty in checking robustness is to apprehend the extra behaviors due to the relaxed model w.r.t. the strong model. This requires a priori reasoning about complex order constraints between operations in arbitrarily long computations, which may need maintaining unbounded ordered structures, and make robustness checking hard or even undecidable.

Our first contribution is to show that verifying robustness of transactional programs against snapshot isolation can be reduced in polynomial time to the reachability problem in concurrent programs under sequential consistency (SC). This allows (1) to avoid explicit handling of the snapshots from where transactions read along computations (since this may imply memorizing unbounded information), and (2) to leverage available tools for verifying invariants/reachability problems on concurrent programs. This also implies that the robustness problem is decidable for finite-state programs, PSPACE-complete when the number of sites is fixed, and EXPSPACE-complete otherwise. This is the first result on the decidability and complexity of the problem of verifying robustness in the context of transactional programs. The problem of verifying robustness has been considered in the literature for several models, including eventual and causal consistency~\cite{DBLP:conf/concur/0002G16,DBLP:conf/popl/BrutschyD0V17,DBLP:conf/pldi/BrutschyD0V18,DBLP:journals/jacm/CeroneG18,sureshconcur2018}. These works provide (over- or under-)approximate analyses for checking robustness, but none of them provides precise (sound and complete) algorithmic verification methods for solving this problem. 

Based on this reduction, our second contribution is a proof methodology for establishing robustness which builds on Lipton's reduction theory~\cite{DBLP:journals/cacm/Lipton75}. We use the theory of movers to establish whether the relaxations allowed by SI are harmless, i.e., they don't introduce new behaviors compared to serializability.

We applied the proposed verification techniques on 10 challenging applications extracted from previous work~\cite{DBLP:conf/pldi/BrutschyD0V18,TPCC,DBLP:conf/icde/AlomariCFR08,DBLP:conf/cloud/HoltBZPOC16,DBLP:conf/concur/0002G16,DBLP:conf/popl/GotsmanYFNS16,DBLP:conf/concur/NagarJ18}. These techniques were enough for proving or disproving the robustness of these applications.


\setlength{\abovedisplayskip}{3pt}
\vspace{-5mm}
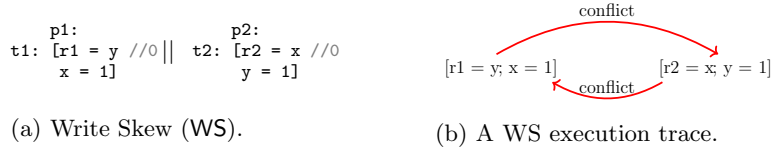
\begin{figure}[t]
\lstset{basicstyle=\ttfamily\scriptsize}

\begin{subfigure}{45mm}
\begin{minipage}[c]{19mm}
\begin{lstlisting}[language=Java10]
    p1:
t1: [r1 = y //0
     x = 1]
\end{lstlisting}
\end{minipage}
\begin{minipage}[c]{3mm}
\footnotesize{$||$}
\end{minipage}
\begin{minipage}[c]{17mm}
\begin{lstlisting}[language=Java10]
    p2:
t2: [r2 = x //0
     y = 1]
\end{lstlisting}
\end{minipage}
\caption{Write Skew ($\mathsf{WS}$).}
\label{fig:rob0}
\end{subfigure}
\hspace{10mm}
\begin{subfigure}{73mm}
\scalebox{0.57}
{
\begin{tikzpicture}

 \node[shape=rectangle ,draw=none,font=\large] (A) at (0,0)  []
 {[r1 = y};
 \node[shape=rectangle ,draw=none,font=\large] (C) at (1.3,0)  []
 {; x = 1]};
  \node[shape=rectangle ,draw=none,font=\large] (B) at (5,0)  [] {[r2 = x};
  \node[shape=rectangle ,draw=none,font=\large] (D) at (6.3,0)  [] {; y = 1]};

  \begin{scope}[ every edge/.style={draw=red,very thick}]
  \path [->] (A) edge [bend left] node [above,font=\large] {conflict} (D);
  \path [->] (B) edge [bend left] node [above,font=\large] {conflict} (C);
  \end{scope}

\end{tikzpicture}}
\caption{A WS execution trace.}
\label{fig:rob0trace}
\end{subfigure}



\caption{Examples of non-robust programs illustrating the difference between \sic{} and serializability.
\emph{causal dependency} means that a read in a transaction obtains its value from a write in another transaction.
\emph{conflict} means that a write in a transaction is not visible to a read in another transaction, but it would affect the read value if it were visible.
Here, \emph{happens-before} is the union of the two.}
\vspace{-3mm}
\end{figure}

\section{Overview}\label{sec:overview}

In this section, we give an overview of our approach for checking robustness against snapshot isolation.
While serializability enforces that transactions are atomic and conflicting transactions, i.e., which read or write to a common location, \emph{cannot} commit concurrently,  \sic{}~\cite{DBLP:conf/sigmod/BerensonBGMOO95} allows that conflicting transactions commit in parallel as long as they don't contain a write-write conflict, i.e., write on a common location. Moreover, under \sic{}, each transaction reads from a snapshot of the database taken at its start.
These relaxations permit the ``anomaly'' known as Write Skew (WS) shown in Figure~\ref{fig:rob0}, where an anomaly is a program execution which is allowed by \sic{}, but not by serializability.
The execution of Write Skew under \sic{} allows the reads of ${\tt x}$ and ${\tt y}$ to return $0$ although this cannot happen under serializability. These values are possible since each transaction is executed locally (starting from the initial snapshot) without observing the writes of the other transaction.

\noindent
{\bf Execution trace.} Our notion of program robustness is based on an abstract representation of executions called \emph{trace}.
Informally, an execution trace is a set of events, i.e., accesses to shared variables and transaction begin/commit events, along with several standard dependency relations between events recording the data-flow. The transitive closure of the union of all these dependency relations is called \emph{happens-before}.
An execution is an anomaly if the happens-before of its trace is cyclic.
Figure~\ref{fig:rob0trace} shows the happens-before of the Write Skew anomaly. Notice that the happens-before order is cyclic in both cases.

Semantically, every transaction execution involves two main events, the issue and the commit. The issue event corresponds to a sequence of reads and/or writes where the writes are visible only to the current transaction. We interpret it as a single event since a transaction starts with a database snapshot that it updates in isolation, without observing other concurrently executing transactions.
The commit event is where the writes are propagated and made visible to all processes. Under serializability, the two events coincide, i.e., they are adjacent in the execution. Under \sic{}, this is not the case and in between the issue and the commit of the same transaction, we may have issue/commit events from concurrent transactions. When a transaction commit does not occur immediately after its issue, we say that the underlying transaction is \emph{delayed}.  For example, the following execution of WS corresponds to the happens-before cycle in Figure~\ref{fig:rob0trace} where the write to $x$ was committed after $\atr_2$ finished, hence, $\atr_1$ was delayed:

{\scriptsize
\begin{align*}
\beginact(\apr_1,\atr_1)\loadact(\apr_1,\atr_1,y,0)\issueact(\apr_1,\atr_1,x,1)&&\hspace{-4mm}\commitact(\apr_1,\atr_1) \\
&\beginact(\apr_2,\atr_2)\loadact(\apr_2,\atr_2,x,0)\issueact(\apr_2,\atr_2,y,1)\commitact(\apr_2,\atr_2)
\end{align*}
}%
Above, $\beginact(\apr_1,\atr_1)$ stands for starting a new transaction $\atr_1$ by process $\apr_1$, $\loadact$ represents read (load) actions, while $\issueact$ denotes write actions that are visible only to the current transaction (not yet committed). The writes performed during $\atr_1$ become visible to all processes once the commit event $\commitact(\apr_1,\atr_1)$ takes place.

\noindent
{\bf Reducing robustness to SC reachability.} The above \sic{} execution can be mimicked by an execution of the same program under serializability modulo an instrumentation that simulates the delayed transaction. The local writes in the issue event are simulated by writes to auxiliary registers and the commit event is replaced by copying the values from the auxiliary registers to the shared variables (actually, it is not necessary to simulate the commit event; we include it here for presentation reasons). The auxiliary registers are visible only to the delayed transaction. In order that the execution be an anomaly (i.e., not possible under serializability without the instrumentation) it is required that the issue and the commit events of the delayed transaction are linked by a chain of happens-before dependencies.
For instance, the above execution for WS can be simulated by:

{\scriptsize
\begin{align*}
\beginact(\apr_1,\atr_1)\loadact(\apr_1,\atr_1,y,0)\writeact(\apr_1,\atr_1,r_{x},1)&&\hspace{-6mm}\writeact(\apr_1,\atr_1,x,r_{x}) \\
&\hspace{-3mm}\beginact(\apr_2,\atr_2)\loadact(\apr_2,\atr_2,x,0)\issueact(\apr_2,\atr_2,y,1)\commitact(\apr_2,\atr_2)
\end{align*}
}%
The write to $x$ was delayed by storing the value in the auxiliary register $r_{x}$  and the happens-before chain exists because the read on $y$ that was done by $\atr_1$ is conflicting with the write on $y$ from $\atr_2$ and the read on $x$ by $\atr_2$ is conflicting with the write of $x$ in the simulation of $t_1$'s commit event. On the other hand, the following execution of Write-Skew without the read on $y$ in $\atr_1$:

{\scriptsize
\begin{align*}
\beginact(\apr_1,\atr_1)\writeact(\apr_1,\atr_1,r_{x},1)&&\hspace{-4mm}\writeact(\apr_1,\atr_1,x,r_{x}) \\
&\beginact(\apr_2,\atr_2)\loadact(\apr_2,\atr_2,x,0)\issueact(\apr_2,\atr_2,y,1)\commitact(\apr_2,\atr_2)
\end{align*}
}%
\noindent
misses the conflict (happens-before dependency) between the issue event of $\atr_1$ and $\atr_2$. Therefore, the events of $\atr_2$ can be reordered to the left of $\atr_1$ and obtain an equivalent execution where $\writeact(\apr_1,\atr_1,x,r_{x})$ occurs immediately after $\writeact(\apr_1,\atr_1,r_{x},1)$. In this case, $\atr_1$ is not anymore delayed and this execution is possible under serializability (without the instrumentation).

If the number of transactions to be delayed in order to expose an anomaly is unbounded, the instrumentation described above may need an unbounded number of auxiliary registers. This would make the  verification problem hard or even undecidable. However,
we show that it is actually enough to delay a single transaction, i.e., a program admits an anomaly under \sic{} iff it admits an anomaly containing a single delayed transaction. This result implies that the number of auxiliary registers needed by the instrumentation is bounded by the number of program variables, and that checking robustness against \sic{} can be reduced in linear time to a reachability problem under serializability (the reachability problem encodes the existence of the chain of happens-before dependencies mentioned above). The proof of this reduction relies on a non-trivial characterization of anomalies.

\noindent
{\bf Proving robustness using commutativity dependency graphs.}  Based on the reduction above, we also devise an approximated method for checking robustness based on the concept of mover in Lipton's reduction theory~\cite{DBLP:journals/cacm/Lipton75}.
 An event is a left (resp., right) mover if it commutes to the left (resp., right) of another event (from a different process) while preserving the computation. We use the notion of mover to characterize happens-before dependencies between transactions. Roughly, there exists a happens-before dependency between two transactions in some execution if one doesn't commute to the left/right of the other one.
 We define a commutativity dependency graph which summarizes the happens-before dependencies in all executions of a given program between transactions $\atr$ as they appear in the program, transactions $\atr\setminus\{w\}$ where the writes of $\atr$ are
\begin{wrapfigure}{r}{5.5cm}
\vspace{-20pt}
\lstset{basicstyle=\ttfamily\scriptsize}
\begin{subfigure}{55mm}
\scalebox{0.57}
{\begin{tikzpicture}
 \node[shape=rectangle ,draw=none,font=\large] (A) at (0,0)  [] {$\atr_1$};
 \node[shape=rectangle ,draw=none,font=\large] (B) at (0,1.5)  [] {$\atr_2$};
  \node[shape=rectangle ,draw=none,font=\large] (C) at (4,0)  [] {$\atr_1\setminus\{r\}$};
    \node[shape=rectangle ,draw=none,font=\large] (D) at (4,1.5)  [] {$\atr_2\setminus\{r\}$};
      \node[shape=rectangle ,draw=none,font=\large] (E) at (8,0)  [] {$\atr_1\setminus\{w\}$};
       \node[shape=rectangle ,draw=none,font=\large] (F) at (8,1.5)  [] {$\atr_2\setminus\{w\}$};
  \begin{scope}[ every edge/.style={draw=black,very thick}]
  \path [->] (C) edge  [bend left] node [above,font=\large] {} (B);
  \path [->] (C) edge  [bend left] node [above,font=\large] {} (F);
  \path [->] (A) edge  [bend left] node [above,font=\large] {} (B);
  \path [->] (A) edge  [bend left] node [above,font=\large] {} (F);
  \path [->] (B) edge  [bend left] node [above,font=\large] {} (C);
  \path [->] (B) edge  [bend left] node [above,font=\large] {} (A);
  \path [->] (F) edge  [bend left] node [above,font=\large] {} (C);
  \path [->] (F) edge  [bend left] node [above,font=\large] {} (A);
  \end{scope}
\end{tikzpicture}}
\end{subfigure}
\vspace{-15pt}
\caption{Commutativity dependency graph of WS where the read of $y$ is omitted.}
\vspace{-20pt}
\label{fig:cdgWS}
\end{wrapfigure}
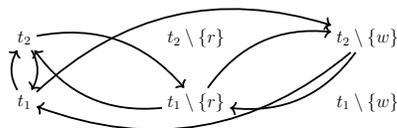
 deactivated (i.e., their effects are not visible outside the transaction), and transactions $\atr\setminus\{r\}$ where the reads of $\atr$ obtain non-deterministic values. The transactions $\atr\setminus\{w\}$ are used to simulate issue events of delayed transactions (where writes are not yet visible) while the transactions $\atr\setminus\{r\}$ are used to simulate commit events of delayed transactions (which only write to the shared memory).
Two transactions $a$ and $b$ are linked by an edge iff $a$ \emph{cannot} move to the right of $b$ (or $b$ cannot move to the left of $a$), or if they are related by the program order (i.e., issued in some order in the same process).
Then a program is robust if for every transaction $\atr$, this graph \emph{doesn't} contain a path from $\atr\setminus\{w\}$ to $\atr\setminus\{r\}$ formed of transactions that don't write to a variable that $\atr$ writes to (the latter condition is enforced by \sic{} since two concurrent transactions cannot commit at the same time when they write to a common variable).
 For example, Figure \ref{fig:cdgWS} shows the  commutativity dependency graph of the modified WS program where the read of $y$ is removed from $\atr_1$. The fact that it doesn't contain any path like above implies that it is robust.

\vspace{-1mm}
\section{Programs}\label{sec:programs}


A program is parallel composition of \emph{processes} distinguished using a set of identifiers $\mathbb{P}$.
Each process is a sequence of \emph{transactions} and each transaction is a sequence of \emph{labeled instructions}.
Each transaction starts with a \plog{begin} instruction and finishes with a \plog{commit} instruction.
Each other instruction is either an assignment to a process-local \emph{register} from a set $\mathbb{R}$ or to a \emph{shared variable} from a set $\mathbb{V}$, or an \plog{assume} statement.
The read/write assignments use values from a data domain $\mathbb{D}$.
An assignment to a register $\langle reg\rangle := \langle var\rangle$ is called a \emph{read} of the shared-variable $\langle var\rangle$ and an assignment to a shared variable $\langle var\rangle := \langle reg\text{-}expr\rangle$ is called a \emph{write} to $\langle var\rangle$ ($\langle reg\text{-}expr\rangle$ is an expression over registers whose syntax we leave unspecified since it is irrelevant for our development).
The \plog{assume} $\langle bexpr\rangle$ blocks the process if the Boolean expression $\langle bexpr\rangle$ over registers is false. They are used to model conditionals as usual.
We use \plog{goto} statements to model an arbitrary control-flow where the same label can be assigned to multiple instructions and multiple \plog{goto} statements can direct the control to the same label which allows to mimic imperative constructs like loops and conditionals. To simplify the technical exposition, our syntax includes simple read/write instructions. However, our results apply as well to instructions that include SQL (select/update) queries. The experiments reported in Section~\ref{sec:Exper:paper} consider programs with SQL based transactions.

%
%
%
%
%
%


The semantics of a program under \sic{} is defined as follows.
The shared variables are stored in a central memory and each process keeps a replicated copy of the central memory.
A process starts a transaction by discarding its local copy and fetching the values of the shared variables from the central memory. When a process commits a transaction, it merges its local copy of the shared  variables with the one stored in the central memory in order to make its updates visible to all processes. During the execution of a transaction, the process stores the writes to shared variables only in its local copy and reads only from its local copy.
When a process merges its local copy with the centralized one, it is required that there were no concurrent updates that occurred after the last fetch from the central memory to a shared variable that was updated by the current transaction. Otherwise, the transaction is aborted and its effects discarded.

More precisely, the semantics of a program $\aprog$ under \sic{} is defined as a labeled transition system $[\aprog]_{\sic{}}$ where transactions are labeled by the set of events
\begin{align*}
\eventsconf = \{\beginact(\apr,\atr), \loadact(\apr,\atr,\anaddr,\aval), \issueact(\apr,\atr,\anaddr,\aval), \commitact(\apr,\atr): \apr\in \mathbb{P}, \atr\in \mathbb{T}^2, \anaddr\in \mathbb{V}, \aval\in \mathbb{D}\}
\end{align*}
where $\beginact$ and $\commitact$ label transitions corresponding to the start and the commit of a transaction, respectively.
$\issueact$ and $\loadact$ label transitions corresponding to writing, resp., reading, a shared variable during some transaction. 


An execution of program $\aprog$, under snapshot isolation, is a sequence of events $\event_1\cdot\event_2\cdot\ldots$ corresponding to a run of $[\aprog]_{\scct{}}$.
 The set of executions of $\aprog$ under \sic{} is denoted by $\executionsconf_{\sic{}}(\aprog)$.




\section{Robustness Against \sic{}}
A \emph{trace} abstracts the order in which shared-variables are accessed inside a transaction and the order between transactions accessing different variables. Formally, the trace of an execution $\rho$ is obtained by (1) replacing each sub-sequence of transitions in $\rho$ corresponding to the same transaction, but excluding the $\commitact$ transition, with a single ``macro-event'' $\issueact(\apr,\atr)$, and (2) adding several standard relations between these macro-events $\issueact(\apr,\atr)$ and commit events $\commitact(\apr,\atr)$ to record the data-flow in $\rho$, e.g. which transaction wrote the value read by another transaction. The sequence of $\issueact(\apr,\atr)$ and $\commitact(\apr,\atr)$ events obtained in the first step is called a \emph{summary of $\rho$}. We say that a transaction $\atr$ in $\rho$ performs an \emph{external read} of a variable $\anaddr$ if $\rho$ contains an event $\loadact(\apr,\atr,\anaddr,\aval)$ which is not preceded by a write on $\anaddr$ of $\atr$, i.e., an event $\issueact(\apr,\atr,\anaddr,\aval)$. Also, we say that a transaction $\atr$ \emph{writes} a variable $\anaddr$ if $\rho$ contains an event $\issueact(\apr,\atr,\anaddr,\aval)$, for some $\aval$.

The \emph{trace} $\traceof{\rho} = (\tau, \po, \rfo, \sto, \cfo, \sametro)$ of an execution $\rho$ consists of the summary $\tau$ of $\rho$ along with the \emph{program order} $\po$, which relates any two issue events $\issueact(\apr,\atr)$ and $\issueact(\apr,\atr')$ that occur in this order in $\tau$, \emph{write-read} relation $\rfo$ (also called \emph{read-from}), which relates any two events $\commitact(\apr,\atr)$ and $\issueact(\apr',\atr')$ that occur in this order in $\tau$ such that $\atr'$ performs an external read of $\anaddr$, and $\commitact(\apr,\atr)$ is the last event in $\tau$ before $\issueact(\apr',\atr')$ that writes to $\anaddr$ (to mark the variable $\anaddr$, we may use $\rfo(x)$), the \emph{write-write} order $\sto$ (also called store-order), which relates any two store events $\commitact(\apr,\atr)$ and $\commitact(\apr',\atr')$ that occur in this order in $\tau$ and write to the same variable $\anaddr$ (to mark the variable $\anaddr$, we may use $\sto(x)$), the \emph{read-write} relation $\cfo$ (also called \emph{conflict}), which relates any two events $\issueact(\apr,\atr)$ and $\commitact(\apr',\atr')$ that occur in this order in $\tau$ such that $\atr$ reads a value that is overwritten by $\atr'$, and the \emph{same-transaction} relation $\sametro$, which relates the issue event with the commit event of the same transaction.
The read-write relation $\cfo$ is formally defined as $\cfo(\anaddr)=\rfo^{-1}(\anaddr);\sto(\anaddr)$ (we use $;$ to denote the standard composition of relations) and $\cfo=\bigcup_{\anaddr\in\mathbb{V}} \cfo(x)$. If a transaction $\atr$ reads the initial value of $\anaddr$ then $\cfo(\anaddr)$ relates $\issueact(\apr,\atr)$ to $\commitact(\apr',\atr')$ of any other transaction $\atr'$ which writes to $\anaddr$ (i.e., $(\issueact(\apr,\atr),\commitact(\apr',\atr'))\in\cfo(\anaddr)$) (note that in the above relations, $\apr$ and $\apr'$ might designate the same process).

Since we reason about only one trace at a time, to simplify the writing, we may say that a trace is simply a sequence $\tau$ as above, keeping the relations  $\po$, $\rfo$, $\sto$, $\cfo$, and $\sametro$ implicit.
The set of traces of executions of a program $\aprog$ under \sic{} is denoted by $\tracesconf(\aprog)_{\sic{}}$.

\noindent
{\bf Serializability semantics.} The semantics of a program under serializability can be defined using a transition system where the configurations keep a single shared-variable valuation (accessed by all processes) with the standard interpretation of read and write statements. Each transaction executes in isolation. Alternatively, the serializability semantics can be defined as a restriction of $[\aprog]_{\sic{}}$ to the set of executions where each transaction is \emph{immediately} delivered when it starts, i.e., the start and commit time of transaction coincide $\atr.\asti=\atr.\acti$. Such executions are called \emph{serializable} and the set of serializable executions of a program $\aprog$ is denoted by $\executionsconf_{\serc{}}(\aprog)$. The latter definition is easier to reason about when relating executions under snapshot isolation and serializability, respectively.

\noindent
{\bf Serializable trace.} A trace $\atrace$ is called \emph{serializable} if it is the trace of a serializable execution. Let $\tracesconf_{\serc{}}(\aprog)$ denote the set of serializable traces. Given a serializable trace $\atrace=(\tau, \po, \rfo, \sto,\cfo, \sametro)$ we have that every event $\issueact(\apr,\atr)$ in $\tau$ is immediately followed by the corresponding $\commitact(\apr,\atr)$ event.

\noindent
{\bf Happens before order.} Since multiple executions may have the same trace, it is possible that an execution $\rho$ produced by snapshot isolation has a serializable trace $\traceof{\rho}$ even though $\issueact(\apr,\atr)$ events may not be immediately followed by $\commitact(\apr,\atr)$ actions. However, $\rho$ would be equivalent, up to reordering of ``independent'' (or commutative) transitions, to a serializable execution. To check whether the trace of an execution is serializable, we introduce the \emph{happens-before} relation on the events of a given trace
as the transitive closure of the union of all the relations in the trace, i.e., $\hbo = (\po \cup \sto \cup \rfo \cup \cfo \cup \sametro)^{+}$.

Finally, the happens-before relation between events is extended to transactions as follows: a transaction $\atr_1$ \emph{happens-before} another transaction $\atr_2\neq \atr_1$ if the trace $\atrace$ contains an event of transaction $\atr_1$ which happens-before an event of $\atr_2$. The happens-before relation between transactions is denoted by $\hbo_t$ and called \emph{transactional happens-before}. 
The following characterizes serializable traces.
\vspace{-0.7mm}
\begin{theorem}[\cite{Adya99,DBLP:journals/toplas/ShashaS88}]\label{th:acyclicity}
 A trace $\atrace$ is serializable iff $\hbo_t$ is acyclic.
\end{theorem}
\vspace{-0.5mm}


A program is called robust if it produces the same set of traces as the serializability semantics.
\vspace{-0.7mm}
\begin{definition}
A program $\aprog$ is called \emph{robust} against \sic{} iff $\tracesconf_{\textsf{\sic{}}}(\aprog)=\tracesconf_{\serc{}}(\aprog)$.
\end{definition}
\vspace{-0.5mm}
Since $\tracesconf_{\serc{}}(\aprog) \subseteq \tracesconf_{\textsf{X}}(\aprog)$, the problem of checking robustness of a program $\aprog$ is reduced to checking whether there exists a trace $\atrace \in \tracesconf_{\sic{}}(\aprog) \setminus \tracesconf_{\serc{}}(\aprog).$
\section{Reducing Robustness against \sic{} to SC Reachability}\label{sec:Instr}
A trace which is not serializable must contain at least an issue and a commit event of the same transaction that don't occur
one after the other even after reordering of ``independent'' events. Thus, there must exist an event that occur between the two which is
 related to both events via the happens-before relation, forbidding the issue and commit to be adjacent. Otherwise, we can build another trace with the same happens-before where events are reordered such that the issue is immediately followed by the corresponding commit.
 The latter is a serializable trace which contradicts the initial assumption.
 We define a program instrumentation which mimics the delay of transactions by doing the writes on auxiliary variables which are not visible to other transactions. After the delay of a transaction, we track happens-before dependencies until we execute a transaction that does a ``read'' on one of the variables that the delayed transaction writes to (this would expose a read-write dependency to the commit event of the delayed transaction). While tracking happens-before dependencies we cannot execute a transaction that writes to a variable that the delayed transaction writes to since \sic{} forbids write-write conflicts between concurrent transactions.

Concretely, given a program $\aprog$, we define an instrumentation of $\aprog$ such that $\aprog$ is not robust against \sic{} iff the instrumentation reaches an error state under serializability.
The instrumentation uses auxiliary variables in order to simulate a \emph{single} delayed transaction which we prove that it is enough for deciding robustness.
Let $\issueact(\apr,\atr)$ be the issue event of the only delayed transaction.
The process $\apr$ that delayed $\atr$ is called the \emph{Attacker}. When the attacker finishes executing the delayed transaction it stops.
Other processes that execute transactions afterwards are called \emph{Happens-Before Helpers}.

The instrumentation uses two copies of the set of shared variables in the original program to simulate the delayed transaction. We use primed variables $\anaddr'$ to denote the second copy. Thus, when a process becomes the attacker, it will only write to the second copy that is not visible to other processes including the happens-before helpers. The writes made by the other processes including the happens-before helpers are made visible to all processes.

When the attacker delays the transaction $\atr$, it keeps track of the variables it accessed, in particular, it stores the name of one of the variables it writes to, $x$, it tracks every variable $y$ that it reads from and every variable $z$ that it writes to.
When the attacker finishes executing $\atr$, and some other process wants to execute some other transaction, the underlying transaction must contain a write to a variable $y$ that the attacker reads from. Also, the underlying transaction must not write to a variable that $\atr$ writes to.
We say that this process has joined happens-before helpers through the underlying transaction.
While executing this transaction, we keep track of each variable that was accessed and the type of operation, wheather it is a read or write.
Afterward, in order for some other transaction to ``join'' the happens-before path, it must not write to a variable that $\atr$ writes to so it does not violate the fact that \sic{} forbids write-write conflicts, and it has to satisfy one of the following conditions in order to ensure the continuity of the  happens-before dependencies:
(1) the transaction is issued by a process that has already another transaction in the  happens-before dependency (program order dependency),
(2) the transaction is reading from a shared variable that was updated by a previous transaction in the happens-before dependency (write-read dependency),
(3) the transaction writes to a shared variable that was read by a previous transaction in the happens-before dependency (read-write dependency), or
(4) the transaction writes to a shared variable that was updated by a previous transaction in the happens-before dependency (write-write dependency).
We introduce a flag for each shared variable to mark the fact that the variable was read or written by a previous transaction. 

Processes continue executing transactions as part of the chain of happens-before dependencies, until a transaction does a read on the variable $x$ that $\atr$ wrote to. In this case, we reached an error state which signals that we found a cycle in the transactional happens-before relation. 

The instrumentation uses four varieties of flags: a) global flags (i.e., $\hb{}$, $a_{\attacktransaction}$, $a_{\attackstore}$), b) flags local to a process (i.e., $p.a$ and $p.hbh$), and c) flags per shared variable (i.e., $\anaddr.event$, $\anaddr.event'$, and $\anaddr.eventI$). We will explain the meaning of these flags along with the instrumentation. At the start of the execution, all flags are initialized to null ($\perp$).

Whether a process is an attacker or happens-before helper is not enforced syntactically by the instrumentation. It is set non-deterministically during the execution using some additional process-local flags. Each process chooses to set to $\mytrue$ at most one of the flags $p.a$ and $p.hbh$, implying that the process becomes an attacker or happens-before helper, respectively. At most one process can be an attacker, i.e., set $p.a$ to $\mytrue$.
In the following, we detail the instrumentation for read and write instructions of the attacker and happens-before helpers.
\subsection{Instrumentation of the Attacker}\label{subsec:Attacker}
Figure \ref{Figure:Attacker} lists the instrumentation of the write and read instructions of the attacker. Each process passes through an initial phase where it executes transactions that are visible immediately to all the other processes (i.e., they are not delayed), and then non-deterministically it can choose to delay a transaction at which point it sets the flag
$a_{\attacktransaction}$ to $\mytrue$. During the delayed transaction it chooses non-deterministically a write instruction to a variable $x$ and stores the name of this variable in the flag  $a_{\attackstore}$ (line \eqref{Equation:AttackerStore}). The values written during the delayed transaction are stored in the primed variables and are visible only to the current transaction, in case the transaction reads its own writes. For example, given a variable $z$, all writes to $z$ from the original program are transformed into writes to the primed version $z'$ (line \eqref{Equation:AttackerDelayStore}). Each time, the attacker writes to $z$, it sets the flag $z.event'=1$. This flag is used later by transactions from happens-before helpers to avoid writing to variables that the delayed transaction writes to.

A read on a variable, $y$, in the delayed transaction takes her value from the primed version, $y'$.
In every read in the delayed transaction, we set the flag $y.event$ to $\loadacc$ (line \eqref{Equation:AttackerLastLoad1}) to be used latter in order for a process to join the happens-before helpers.
Afterward, the attacker starts the happens-before path,  and it sets the variable $\hb{}$ to $\mytrue$ (line \eqref{Equation:AttackerLastLoad}) to mark the start of the happens.
When the flag $\hb{}$ is set to $\mytrue$ the attacker stops executing new transactions. 

\begin{figure}[t]
 \scriptsize
 \begin{minipage}{0.45\linewidth}
\begin{eqnarray}
&&\semidler{\thetransition{\alab_1}{\alab_2}{\theload{\areg}{\anaddr}}}  =\notag\\
&&\textbf{// Read before the delayed transaction}\notag\\
&&\thetransitionInstrumented{\alab_1}{\alab_{x1}}{\thecondition{ a_{\attacktransaction} = \perp}}\notag\\
&&\thetransitionInstrumented{\alab_{x1}}{\alab_{2}}{\theload{\areg}{\anaddr}}\notag\\
&&\textbf{// Read in the delayed transaction}\notag\\
&&\thetransitionInstrumented{\alab_1}{\alab_{x2}}{\thecondition{ a_{\attacktransaction} \neq \perp \land \apr.a\neq \perp}}\notag\\
&&\thetransitionInstrumented{\alab_{x2}}{\alab_{x3}}{\theload{\areg}{\anaddr'}}\notag\\
&&\thetransitionInstrumented{\alab_{x3}}{\alab_{x4}}{\theassign{\anaddr.event}{\loadacc}}\label{Equation:AttackerLastLoad1}\\
&&\thetransitionInstrumented{\alab_{x4}}{\alab_{x5}}{\thecondition{\hb{} = \perp}}\notag\\
&&\thetransitionInstrumented{\alab_{x5}}{\alab_{2}}{\thestore{\hb{}}{\mytrue}}\label{Equation:AttackerLastLoad}\\
&&\thetransitionInstrumented{\alab_{x4}}{\alab_{2}}{\thecondition{\hb{} \neq \perp}}\notag
\end{eqnarray}
 \end{minipage}\hfill
 \begin{minipage}{0.53\linewidth}
 \begin{eqnarray}
&&\semidler{\thetransition{\alab_1}{\alab_2}{\thestore{\anaddr}{e}}} = \notag\\
&&\textbf{// Write before the delayed transaction}\notag\\
&&\thetransitionInstrumented{\alab_1}{\alab_{x1}}{\thecondition{ a_{\attacktransaction} = \perp}}\notag\\
&&\thetransitionInstrumented{\alab_{x1}}{\alab_{2}}{\thestore{\anaddr}{e}}\notag\\
&&\textbf{// Write in the delayed transaction}\notag\\
&&\thetransitionInstrumented{\alab_1}{\alab_{x2}}{\thecondition{ a_{\attacktransaction} \neq \perp \land \apr.a\neq \perp}}\notag\\
&&\thetransitionInstrumented{\alab_{x2}}{\alab_{x3}}{\thestore{\anaddr'}{e}}\label{Equation:AttackerDelayStore}\\
&&\thetransitionInstrumented{\alab_{x3}}{\alab_{2}}{\theassign{\anaddr.event'}{1}}\label{Equation:AttackerDelayStore0}\\
&&\textbf{// Special write in the delayed transaction}\notag\\
&&\thetransitionInstrumented{\alab_1}{\alab_{x4}}{\thecondition{a_{\attackstore} = \anaddr.event = \perp \land \  a_{\attacktransaction} \neq \perp}}\notag\\
&&\thetransitionInstrumented{\alab_{x4}}{\alab_{x5}}{\thestore{\anaddr'}{e}}\notag\\
&&\thetransitionInstrumented{\alab_{x5}}{\alab_{x6}}{\thestore{a_{\attackstore}}{`\anaddr`}}\label{Equation:AttackerStore}\\
&&\thetransitionInstrumented{\alab_{x8}}{\alab_{2}}{\theassign{\anaddr.event'}{1}}\notag
\end{eqnarray}
 \end{minipage}
\normalsize
\vspace{-2mm}
\caption{Instrumentation of the Attacker. We use $`\anaddr`$ to denote the name of the shared variable $\anaddr$.}
\label{Figure:Attacker}
\vspace{-4mm}
\end{figure}
\subsection{Instrumentation of the Happens-Before Helpers}\label{subsec:HBHelpers}
The remaining processes, which are not the attacker, can become a happens-before helper.
Figure \ref{Figure:HBHelpersInstrumentation} lists the instrumentation of write and read instructions of a happens-before helper.
In a first phase, each process executes the original code until the flag $a_{\attacktransaction}$ is set to $\mytrue$ by the attacker. This flag signals the ``creation'' of the secondary copy of the shared-variables, which can be observed only by the attacker. At this point, the flag $\hb{}$ is set to $\mytrue$, and the happens-before helper process
chooses non-deterministically a first transaction through which it wants to join the set of happens-before helpers, i.e., continue the happens-before dependency created by the existing happens-before helpers. When a process chooses a transaction, it makes a pledge (while executing the \plog{begin} instruction) that during this transaction it will either
read from a variable that was written to by another happens-before helper, write to a variable that was accessed (read or written) by another happens-before helper, or write to a variable that was read from in the delayed transaction.
When the pledge is met, the process sets the flag $p.hbh$ to $\hbhjoin$ (lines \eqref{Equation:HBHMeetConditionLoad} and \eqref{Equation:HBHMeetConditionStore}). The execution is blocked if a process does not keep its pledge (i.e., the flag $p.hbh$ is null) at the end of the transaction.
Note that the first process to join the happens-before helper has to execute a transaction $\atr$ which writes to a variable that was read from in the delayed transaction since this is the only way to build a happens-before between $\atr$, and the delayed transaction ($\po$ is not possible since $\atr$ is not from the attacker, $\rfo$ is not possible since $\atr$ does not see the writes of the delayed transaction, and $\sto$ is not possible since $\atr$ cannot write to a variable that the delayed transaction writes to).
We use a flag $x.event$ for each variable $x$ to record the type (read $\loadacc$ or write $\storeacc$) of the last access made by a happens-before helper (lines \eqref{Equation:HBHFlagLoad} and \eqref{Equation:HBHFlagStore}). During the execution of a transaction that is part of the happens-before dependency, we must ensure that the transaction does not write to variable $y$ where $y.even'$ is set to $1$. Otherwise, the execution is blocked (line \ref{Equation:HBHBlockStore}).

The happens-before helpers continue executing their instructions, until one of them reads from the shared variable $x$ whose name was stored in $a_{\attackstore}$. This establishes a happens-before dependency between the delayed transaction and a ``fictitious'' store event corresponding to the delayed transaction that could be executed just after this read of $x$. The execution doesn't have to contain this store event explicitly since it is always enabled. Therefore, at the end of every transaction, the instrumentation checks whether the transaction read $x$. If it is the case, then the execution stops and goes to an error state to indicate that this is a robustness violation.
Notice that after the attacker stops, the only processes that are executing transactions are happens-before helpers, which is justified since when a process is not from a happens-before helper it implies that we cannot construct a happens-before dependency between a transaction of this process and the delayed transaction which means that the two transactions commute which in turn implies that this process's transactions can be executed before executing the delayed transaction of the attacker.

\begin{figure}[t]
 \scriptsize
 \hspace{-3mm}
\begin{minipage}{0.49\linewidth}
\begin{eqnarray}
&&\semHbhelperorig{\thetransition{\alab_1}{\alab_2}{\theload{\areg}{\anaddr}}}=\notag\\[1mm]
&&\textbf{// Read before the delayed transaction}\notag\\
&&\thetransitionInstrumented{\alab_{1}}{\alab_{x1}}{\thecondition{\hb{} = \perp \land p.a = \perp}}\notag\\
&&\thetransitionInstrumented{\alab_{x1}}{\alab_{2}}{\theload{\areg}{\anaddr}}\label{Equation:HBHLoad0}\\
&&\textbf{// Read after the delayed transaction}\notag\\
&&\thetransitionInstrumented{\alab_1}{\alab_{x2}}{\thecondition{ \hb{} \neq \perp}}\notag\\
&&\thetransitionInstrumented{\alab_{x2}}{\alab_{x3}}{\theload{\areg}{\anaddr}}\notag\\
&&\thetransitionInstrumented{\alab_{x3}}{\alab_{x4}}{\thecondition{\anaddr.eventI = \storeacc \land p.hbh = \perp}}\notag\\
&&\thetransitionInstrumented{\alab_{x4}}{\alab_{2}}{\theassign{p.hbh}{\hbhjoin}}\label{Equation:HBHMeetConditionLoad}\\
&&\thetransitionInstrumented{\alab_{x3}}{\alab_{x5}}{\thecondition{\anaddr.event = \perp}}\notag\\
&&\thetransitionInstrumented{\alab_{x5}}{\alab_{2}}{\theassign{\anaddr.event}{\loadacc}}\label{Equation:HBHFlagLoad}\\
&&\thetransitionInstrumented{\alab_{x3}}{\alab_{2}}{\thecondition{\anaddr.event \neq \perp \lor \  p.hbh \neq \perp}}\notag
\end{eqnarray}
 \end{minipage}
\begin{minipage}{0.49\linewidth}
\begin{eqnarray}
&&\semHbhelperorig{\thetransition{\alab_1}{\alab_2}{\thestore{\anaddr}{\anexpr}}}=\notag\\[1mm]
&&\textbf{// Write before the delayed transaction}\notag\\
&&\thetransitionInstrumented{\alab_{1}}{\alab_{x1}}{\thecondition{\hb{} = \perp \land \  a_{\attacktransaction} = \perp}}\notag\\
&&\thetransitionInstrumented{\alab_{x1}}{\alab_{2}}{\thestore{\anaddr}{\anexpr}}\notag\\
&&\textbf{// Write after the delayed transaction}\notag\\
&&\thetransitionInstrumented{\alab_1}{\alab_{x2}}{\thecondition{ \hb{} \neq \perp \land p.a = \perp}}\notag\\
&&\thetransitionInstrumentedb{\alab_{x2}}{\thecondition{\anaddr.event' \neq \perp}}\label{Equation:HBHBlockStore}\\
&&\thetransitionInstrumented{\alab_{x2}}{\alab_{x3}}{\thecondition{\anaddr.event' = \perp}}\notag\\
&&\thetransitionInstrumented{\alab_{x3}}{\alab_{x4}}{\thestore{\anaddr}{\anexpr}}\notag\\
&&\thetransitionInstrumented{\alab_{x4}}{\alab_{x5}}{\theassign{\anaddr.event}{\storeacc}}\label{Equation:HBHFlagStore}\\
&&\thetransitionInstrumented{\alab_{x5}}{\alab_{x6}}{\thecondition{\anaddr.eventI \neq \perp \land \  p.hbh = \perp}}\notag\\
&&\thetransitionInstrumented{\alab_{x6}}{\alab_{2}}{\theassign{p.hbh}{\hbhjoin}}\label{Equation:HBHMeetConditionStore}\\
&&\thetransitionInstrumented{\alab_{x5}}{\alab_{2}}{\thecondition{\anaddr.eventI = \perp \lor \  p.hbh \neq \perp}}\notag
\end{eqnarray}
 \end{minipage}
\normalsize
\vspace{-2mm}
\caption{Instrumentation of Happens-Before Helpers.}
\label{Figure:HBHelpersInstrumentation}
\vspace{-4mm}
\end{figure}
\subsection{Correctness}
The role of a process in an execution is chosen non-deterministically at runtime. Therefore, the final instrumentation of a given program $\aprog$, denoted by $\sem{\aprog}$, is obtained by replacing each labeled instruction $\langle linst\rangle$ with the concatenation of the instrumentations corresponding to the attacker and the happens-before helpers, i.e., $\ \ \ \sem{\langle linst\rangle} ::= \semidler{\langle linst\rangle} \hspace{0.17cm}   \semHbhelperorig{\langle linst\rangle} $

The following theorem states the correctness of the instrumentation.
\vspace{-0.7mm}
\begin{theorem}\label{th:SIfinal}
\mbox{$\aprog$ is not robust against \sic{} iff $\sem{\aprog}$ reaches the error state.}
\end{theorem}
\vspace{-0.7mm}
If a program is not robust, this implies that the execution of the program under \sic{} results in a trace where the happens-before is cyclic.
Which is possible only if the program contains at least one delayed transaction.
In the proof of this theorem, we show that is sufficient to search for executions that contain a single delayed transaction.

Notice that in the instrumentation of the attacker, the delayed transaction must contain a read and write instructions on different variables. Also, the transactions of the happens-before helpers must not contain a write to a variable that the delayed transaction writes to. 
The following corollary states the complexity of checking robustness for finite-state programs\,\footnote{Programs with a bounded number of variables taking values from a bounded domain.} against snapshot isolation. It is a direct consequence of Theorem~\ref{th:SIfinal} and of previous results concerning the reachability problem in concurrent programs running over a sequentially-consistent memory, with a fixed~\cite{DBLP:conf/focs/Kozen77} or parametric number of processes~\cite{DBLP:journals/tcs/Rackoff78}.
\vspace{-1mm}
%
%
\begin{corollary}\label{theorem:InstRobus}
Checking robustness of finite-state programs against snapshot isolation is PSPACE-complete when the number of processes is
fixed and\\ EXPSPACE-complete, otherwise.
\end{corollary}
\vspace{-0.7mm}
The instrumentation can be extended to SQL (select/update) queries where a statement may include expressions over a finite/infinite set of variables, e.g., by manipulating a set of flags x.event for each statement instead of only one. 

\section{Proving Program Robustness}\label{sec:PPR}
As a more pragmatic alternative to the reduction in the previous section, we define an approximated method for proving robustness which is inspired by Lipton's reduction theory~\cite{DBLP:journals/cacm/Lipton75}.

\smallskip
\noindent
{\bf Movers.} Given an execution $\tau = \event_1\cdot \ldots\cdot \event_n$ of a program $\aprog$ under serializability (where each event $\event_i$ corresponds to executing an entire transaction), we say that the event $\event_i$ \emph{moves right (resp., left)} in $\tau$ if $\event_1\cdot \ldots\cdot \event_{i-1}\cdot \event_{i+1}\cdot \event_i\cdot \event_{i+2}\cdot \ldots\cdot \event_n$ (resp., $\event_1\cdot\ldots\cdot\event_{i-2}\cdot\event_i\cdot\event_{i-1}\cdot\event_{i+1}\cdot\ldots\cdot\event_n$) is also a valid execution of $\aprog$, the process of $\event_i$ is different from the process of $\event_{i+1}$ (resp., $\event_{i-1}$), and both executions reach to the same end state $\sigma_n$.
For an execution $\tau$, let $\mathsf{instOf}_\tau(\event_i)$ denote the transaction that generated the event $\event_i$.
A transaction $\atr$ of a program $\aprog$ is a \emph{right (resp., left) mover} if for all executions $\tau$ of $\aprog$ under serializability, the event $\event_i$ with $\mathsf{instOf}(\event_i) = \atr$ moves right (resp., left) in $\tau$.


If a transaction $\atr$ is not a right mover, then there must exist an execution $\tau$ of $\aprog$ under serializability and an event $\event_i$ of $\tau$ with $\mathsf{instOf}(\event_i) = \atr$ that does not move right. This implies that there must exist another $\event_{i+1}$ of $\tau$ which caused $\event_i$ to not be a right mover. Since $\event_i$ and $\event_{i+1}$ do not commute, then this must be because of either a write-read, write-write, or a read-write dependency. If  $\atr'=\mathsf{instOf}(\event_{i+1})$, we say that $\atr$ is not a right mover because of $\atr'$ and some dependency that is either write-read, write-write, or read-write.
Notice that when $\atr$ is not a right mover because of $\atr'$ then $\atr'$ is not a left mover because of $\atr$.

We define $\mrfo$ as a binary relation between transactions such that $(\atr,\atr') \in \mrfo$ when $\atr$ is \emph{not} a right mover because of $\atr'$ and a write-read dependency. We define the relations $\msto$ and $\mcfo$ corresponding to write-write and read-write dependencies in a similar way.

\smallskip
\noindent
{\bf Read/Write-free transactions.}  Given a transaction $\atr$, we define $\atr\setminus\{r\}$ as a variation of $\atr$ where all the reads from shared  variables are replaced with non-deterministic reads, i.e., $\langle reg\rangle := \langle var\rangle$ statements are replaced with $\langle reg\rangle := \star$ where $\star$ denotes non-deterministic choice. We also define $\atr\setminus\{w\}$ as a variation of $\atr$ where all the writes to shared  variables in $\atr$ are disabled. Intuitively, recalling the reduction to SC reachability in Section~\ref{sec:Instr}, $\atr\setminus\{w\}$ simulates the delay of a transaction by the Attacker, i.e., the writes are not made visible to other processes, and $\atr\setminus\{r\}$ approximates the commit of the delayed transaction which only applies a set of writes.

\smallskip
\noindent
{\bf Commutativity dependency graph.} Given a program $\aprog$, we define the commutativity dependency graph as a graph where vertices represent transactions and their read/write-free variations. Two vertices which correspond to the original transactions in $\aprog$ are related by a program order edge, if they belong to the same process. The other edges in this graph represent the ``non-mover'' relations $\mrfo$, $\msto$, and $\mcfo$.
Given a program $\aprog$, we say that the commutativity dependency graph of $\aprog$ contains a \emph{non-mover cycle} if there exist a set of transactions $\atr_0, \atr_1, \ldots, \atr_n$ of $\aprog$ such that the following hold:
 \begin{enumerate}[label=(\alph*), topsep=0pt]
\item $(\atr''_0,\atr_1) \in \mcfo$ where $\atr''_0$ is the write-free variation of $\atr_0$ and $\atr_1$ does not write to a variable that $\atr_0$ writes to;
\item for all $i\in[1,n]$, $(\atr_i,\atr_{i+1})\in (\po \cup \mrfo \cup \msto \cup \mcfo)$, $\atr_i$ and $\atr_{i+1}$ do not write to a shared variable that $\atr_0$ writes to;
\item $(\atr_n,\atr'_0) \in \mcfo$ where $\atr'_0$ is the read-free variation of $\atr_0$ and $\atr_n$ does not write to a variable that $\atr_0$ writes to.
\end{enumerate}

A non-mover cycle approximates an execution of the instrumentation defined in Section~\ref{sec:Instr} in between the moment that the Attacker delays a transaction $\atr_0$ (which here corresponds to the write-free variation $\atr''_0$) and the moment where $\atr_0$ gets committed (the read-free variation $\atr'_0$).


The following theorem shows that the acyclicity of the commutativity dependency graph of a program implies the robustness of the program. Actually, the notion of robustness in this theorem relies on a slightly different notion of trace where store-order and write-order dependencies take into account values, i.e., store-order relates only writes writing different values and the write-order relates a read to the oldest write (w.r.t. execution order) writing its value. This relaxation helps in avoiding some harmless robustness violations due to for instance, two transactions writing the same value to some variable.
\vspace{-0.7mm}
\begin{theorem}\label{soundnessthm}
For a program $\aprog$, if the commutativity dependency graph of $\aprog$ does not contain non-mover cycles, then $\aprog$  is robust.
\end{theorem}
\vspace{-0.5mm}
\vspace{-3mm}
\section{Experiments}\label{sec:Exper:paper}
\vspace{-0.7mm}
To test the applicability of our robustness checking algorithms, we have considered a benchmark of 10 applications extracted from the literature related to weakly consistent databases in general.
A first set of applications are open source projects that were implemented to be run over the Cassandra database, extracted from~\cite{DBLP:conf/pldi/BrutschyD0V18}. The second set of applications is composed of: TPC-C \cite{TPCC}, an on-line transaction processing benchmark widely used in the database community, SmallBank, a simplified representation of a banking application~\cite{DBLP:conf/icde/AlomariCFR08}, FusionTicket, a movie ticketing application~\cite{DBLP:conf/cloud/HoltBZPOC16}, Auction, an online auction application~\cite{DBLP:conf/concur/0002G16}, and Courseware, a course registration service extracted from \cite{DBLP:conf/popl/GotsmanYFNS16,DBLP:conf/concur/NagarJ18}.

A first experiment concerns the reduction of robustness checking to SC reachability. For each application, we have constructed a client (i.e., a program composed of transactions defined within that application) with a fixed number of processes (at most 3) and a fixed number of transactions (between 3 and 7 transactions per process). We have encoded the instrumentation of this client, defined in Section~\ref{sec:Instr}, in the Boogie programming language \cite{DBLP:conf/fmco/BarnettCDJL05} and used the Civl verifier~\cite{DBLP:conf/cav/HawblitzelPQT15} in order to check whether the assertions introduced by the instrumentation are violated (which would represent a robustness violation). Note that since clients are of fixed size, this requires no additional assertions/invariants (it is an instance of bounded model checking). The results are reported in Table~\ref{tab:exper}. We have found two of the applications, Courseware and SmallBank, to \emph{not} be robust against snapshot isolation. The violation in Courseware is caused by transactions RemoveCourse and EnrollStudent that execute concurrently, RemoveCourse removing a course that has no registered student and EnrollStudent registering a new student to the same course. We get an invalid state where a student is registered for a course that was removed. SmallBank's violation contains transactions Balance, TransactSaving, and WriteCheck. One process executes WriteCheck where it withdraws an amount from the checking account after checking that the sum of the checking and savings accounts is bigger than this amount. Concurrently, a second process executes TransactSaving where it withdraws an amount from the saving account after checking that it is smaller than the amount in the savings account. Afterwards, the second process checks the contents of both the checking and saving accounts. We get an invalid state where the sum of the checking and savings accounts is negative.

Since in the first experiment we consider fixed clients, the lack of assertion violations doesn't imply that the application is robust (this instantiation of our reduction can only be used to reveal robustness violations). Thus, a second experiment concerns the robustness proof method based on commutativity dependency graphs (Section~\ref{sec:PPR}). For the applications that were not identified as non-robust by the previous method, we have used Civl to construct their commutativity dependency graphs, i.e., identify the ``non-mover'' relations $\mrfo$, $\msto$, and $\mcfo$ (Civl allows to check whether some code fragment is a left/right mover). In all cases, the graph didn't contain non-mover cycles, which allows to conclude that the applications are robust.

The experiments show that our results can be used for finding violations and proving robustness, and that they apply to a large set of interesting examples. Note that the reduction to SC and the proof method based on commutativity dependency graphs are valid for programs with SQL (select/update) queries.

%
%

\begin{table}[t]
\caption{An overview of the analysis results.
CDG stands for commutativity dependency graph.
The columns PO and PT show the number of proof obligations and proof time in second, respectively.
\textsf{T} stands for trivial when the application has only read-only transactions.
}
\centering
\begin{tabular}{|c|c|c|c|c|c|c|}
\hline
Application & \#Transactions & Robustness & \multicolumn{2}{c|}{Reachability Analysis} &  \multicolumn{2}{c|}{CDG Analysis}\\
 &  &  & PO & PT & PO & PT\\
\hline
Auction & 4 & \cmark & 70 & 0.3 & 20 & 0.5\\
\hline
Courseware &  5 & \xmark &  59 & 0.37 & \textsf{na} & \textsf{na} \\ 
\hline
FusionTicket & 4 & \cmark & 72 & 0.3 & 34 & 0.5\\
\hline
SmallBank & 5 & \xmark & 48 & 0.28 & \textsf{na} & \textsf{na}\\ 
\hline
TPC-C & 5 & \cmark & 54 & 0.7 & 82 & 3.7\\
\hline
\hline
Cassieq-Core &  8 & \cmark & 173 & 0.55 & 104 & 2.9\\
\hline
Currency-Exchange &  6 & \cmark & 88  & 0.35 & 26 & 3.5\\
\hline
PlayList & 14 & \cmark & 99 & 4.63 & 236 & 7.3 \\
\hline
RoomStore & 5 & \cmark & 85 & 0.3 & 22 & 0.5\\
\hline
Shopping-Cart & 4 & \cmark & 58 & 0.25 & \textsf{T} & \textsf{T} \\
\hline
\end{tabular}
\label{tab:exper}
\end{table}
\vspace{-3.7mm} 
\section{Related Work}
Decidability and complexity of robustness has been investigated in the context of relaxed memory models such as TSO and Power \cite{DBLP:conf/icalp/BouajjaniMM11,DBLP:conf/esop/BouajjaniDM13,DBLP:conf/icalp/DerevenetcM14}.
Our work borrows some high-level principles from~\cite{DBLP:conf/esop/BouajjaniDM13} which addresses the robustness against TSO. We reuse the high-level methodology of characterizing minimal violations according to some measure and defining reductions to SC reachability using a program instrumentation. Instantiating this methodology in our context is however very different, several fundamental differences being:
\begin{enumerate}[label={--}, topsep=0pt]
\item \sic{} and TSO admit different sets of relaxations and \sic{} is a model of transactional databases.
\item We use a different notion of measure: the measure in \cite{DBLP:conf/esop/BouajjaniDM13} counts the number of events between a write issue and a write commit while our notion of measure counts the number of delayed transactions. This is a first reason for which the proof techniques in~\cite{DBLP:conf/esop/BouajjaniDM13} don't extend to our context.
\item Transactions induce more complex traces: two transactions might be related by several dependency relations since each transaction may contain multiple reads and writes to different locations. In TSO, each action is a read or a write to some location, and two events are related by a single dependency relation. Also, the number of dependencies between two transactions depends on the execution since the set of reads/writes in a transaction evolves dynamically. 
\end{enumerate}
Other works~\cite{DBLP:conf/icalp/BouajjaniMM11,DBLP:conf/icalp/DerevenetcM14} define decision procedures which are based on the theory of regular languages and do not extend to infinite-state programs like in our case. 

\begin{wrapfigure}{r}{0.43\textwidth}
\vspace{-30pt}
  \begin{center}
    \lstset{basicstyle=\ttfamily\scriptsize}
\begin{minipage}[c]{23mm}
\begin{lstlisting}[language=Java10]
    p1:
t1: [ if (x > y)
        r1 = x - y
        x  = y ]
\end{lstlisting}
\end{minipage}
\begin{minipage}[c]{2mm}
\footnotesize{$||$}
\end{minipage}
\begin{minipage}[c]{23mm}
\begin{lstlisting}[language=Java10]
    p2:
t2: [ if (y > x)
        r2 = y - x
        y  = x ]
\end{lstlisting}
\end{minipage}
  \end{center}
  \vspace{-20pt}
  \caption{A robust program.}
  \vspace{-20pt}
\label{fig:BAB}
\end{wrapfigure}
As far as we know, our work provides the first results concerning the decidability and the complexity of robustness checking in the context of transactions.
The existing work on the verification of robustness for transactional programs provide either over- or under-approximate analyses. Our commutativity dependency graphs are similar to the static dependency graphs used in~\cite{DBLP:conf/concur/0002G16,DBLP:conf/popl/BrutschyD0V17,DBLP:conf/pldi/BrutschyD0V18,DBLP:journals/jacm/CeroneG18}, but they are more precise, i.e., reducing the number of false alarms. The static dependency graphs record happens-before dependencies between transactions based on a syntactic approximation of the variables accessed by a transaction. For example, our techniques are able to prove that the program in Figure \ref{fig:BAB} is robust, while this is not possible using static dependency graphs. The latter would contain a dependency from transaction \emph{$\atr_1$} to \emph{$\atr_2$} and one from \emph{$\atr_2$} to \emph{$\atr_1$} just because syntactically, each of the two transactions reads both variables and may write to one of them. Our dependency graphs take into account the semantics of these transactions and do not include this happens-before cycle.
Other over- and under-approximate analyses have been proposed in~\cite{sureshconcur2018}. They are based on encoding executions into first order logic, bounded-model checking for the under-approximate analysis, and a sound check for proving a cut-off bound on the size of the happens-before cycles possible in the executions of a program, for the over-approximate analysis.
The latter is strictly less precise than our method based on commutativity dependency graphs. For instance, extending the TPC-C application with additional transactions will make the method in~\cite{sureshconcur2018} fail while our method will succeed in proving robustness (the three transactions are for adding a new product, adding a new warehouse based on the number of customers and warehouses, and adding a new customer, respectively).

Finally, the idea of using Lipton's reduction theory for checking robustness has been also used in the context of the TSO memory model~\cite{DBLP:conf/cav/BouajjaniEMT18}, but the techniques are completely different, e.g., the TSO technique considers each update in isolation and doesn't consider non-mover cycles like in our commutativity dependency graphs.

\bibliographystyle{splncs04}
\bibliography{main}

\appendix
\newpage
\section{Programs}\label{app:programs}


\subsection{Program Syntax}

We consider a simple programming language grammar which is defined in Figure~\ref{Figure:syntax}. 
A program is parallel composition of \emph{processes} distinguished using a set of identifiers $\mathbb{P}$. 
Each process is a sequence of \emph{transactions} and each transaction is a sequence of \emph{labeled instructions}. 
Each transaction starts with a \plog{begin} instruction and finishes with a \plog{commit} instruction.
Each other instruction is either an assignment to a process-local \emph{register} from a set $\mathbb{R}$ or to a \emph{shared variable} from a set $\mathbb{V}$, or an \plog{assume} statement.
The read/write assignments use values from a data domain $\mathbb{D}$.
An assignment to a register $\langle reg\rangle := \langle var\rangle$ is called a \emph{read} of $\langle var\rangle$ and an assignment to a shared variable $\langle var\rangle := \langle reg\text{-}expr\rangle$ is called a \emph{write} to $\langle var\rangle$ ($\langle reg\text{-}expr\rangle$ is an expression over registers whose syntax we leave unspecified since it is irrelevant for our development).
The \plog{assume} $\langle bexpr\rangle$ blocks the process if the Boolean expression $\langle bexpr\rangle$ over registers is false.

\begin{figure}
{\footnotesize
\setlength{\grammarindent}{7em}
\begin{grammar}
<prog>  ::= \plog{program} <process>$^{*}$

<process> ::= \plog{process} <pid> \plog{regs} <reg>$^{*}$ \\ <ltxn>$^{*}$

<ltxn> ::=  <binst> <linst>$^{*}$ <einst>

<binst> ::= <label>":" \plog{begin}";" \plog{goto} <label>";"

<einst> ::= <label>":" \plog{end}";" \plog{goto} <label>";"

<linst> ::= <label>":" <inst>";" \plog{goto} <label>";"

<inst> ::= <reg> ":=" <var>
  \alt <var> ":=" <reg-expr>
  \alt \plog{assume} <bexpr>
\end{grammar}}
\vspace{-2mm}
\caption{Program syntax. $a^{*}$ indicates zero or more occurrences of $a$.  $\langle pid\rangle$, $\langle reg\rangle$, $\langle label \rangle$, and $\langle var\rangle$ represent a process identifier, a register, a label, and a shared variable, respectively. $\langle reg\text{-}expr \rangle$ is an expression over registers while $\langle bexpr \rangle$ is a Boolean expression over registers.
}
\label{Figure:syntax}
\end{figure}

\subsection{Program Semantics Under SnapShot Isolation}

The semantics of a program under \sic{} is defined as follows. 
The shared variables are stored in a central memory and each process keeps a replicated copy of the central memory. 
A process starts a transaction by discarding its local copy and fetching the values of the shared variables from the central memory. When a process commits a transaction, it merges its local copy of the shared  variables with the one stored in the central memory in order to make its updates visible to all processes. During the execution of a transaction, the process stores the writes to shared variables only in its local copy and reads only from its local copy.
When a process merges its local copy with the centralized one, it is required that there were no concurrent updates that occurred after the last fetch from the central memory to a shared variable that was updated by the current transaction. Otherwise, the transaction is aborted and its effects discarded. 

Thus, a program configuration is a tuple $\gsconf = (\lsconf,\timest,\logconf)$ where $\lsconf: \mathbb{P} \rightarrow \lstatesconf$ associates a local state in $\lstatesconf$ to each process in $\mathbb{P}$, $\timest: \mathbb{V} \rightarrow \mathbb{T}$ stores the largest timestamp for each shared variable, and $\logconf: \mathbb{V} \rightarrow \mathbb{D}$ holds the global valuation of shared variables.
A local state is a tuple  $\tuple{\pcconf,\storeconf,\txnwrsconf,\valconf}$ where
$\pcconf \in\labconf$ is the program counter, i.e., the label of the next instruction to be executed, $\storeconf: \mathbb{V} \rightarrow \mathbb{D}$ is the local valuation of the shared variables, $\txnwrsconf: \mathbb{V} \rightarrow \{\bot,1\}$ is a local log which marks shared variables which were updated in a transaction, and $\valconf: \mathbb{R} \rightarrow \mathbb{D}$ is the valuation of the local registers. For a local state $s$, we use $s.\pcconf$ to denote the program counter component of $s$, and similarly for all the other components of $s$. Given a transaction $\atr \in \mathbb{T} \times \mathbb{T}$, we use $\atr.\asti$ to denote the start time of transaction $\atr$ and $\atr.\acti$ to denote the commit time of $\atr$.
Before merging $\storeconf$ with $\logconf$, after executing a transaction $\atr$, we check that for every variable $x$ that $\atr$ writes to ($\txnwrsconf(x) \neq \bot$) we have that $\timest(x) < \atr.\asti$ (i.e., there were no concurrent write to $x$). Then, we store the value of $\storeconf(x)$ for every variable $x$ that $\atr$ writes to ($\txnwrsconf(x) \neq \bot$) in $\logconf(x)$. Also, for every variable $x$ that $\atr$ writes to, we store $\atr.\acti$ in $\timest(x)$.

\begin{figure}[t]
\centering
\small\addtolength{\tabcolsep}{-10pt}
\therules{
\therule
{$\text{\plog{begin}}\in\instrOf(\lsconf(\apr).\pcconf)$\quad $\mathsf{img}(\lsconf.\timest) < \atr.\asti$\quad $s = \lsconf(\apr)[\txnwrsconf \mapsto \epsilon,\storeconf\mapsto\logconf,\pcconf\mapsto \mathsf{next}(\pcconf)]$}
{$(\lsconf,\timest,\logconf,\lkconf)\mpitrans{\beginact(\apr,\atr)} (\lsconf[\apr\mapsto s],\timest,\logconf,\lkconf)$}

\dfrac
{\splitdfrac{\text{$\theload{\areg}{\anaddr}\in\instrOf(\lsconf(\apr).\pcconf)$\quad
$\lsconf(\apr).\storeconf[\anaddr] = \aval$ \quad
$\mathit{rval} = \lsconf(\apr).\valconf[\areg\mapsto \aval]$}}
{\text{$s = \lsconf(\apr)[\valconf \mapsto \mathit{rval},\pcconf\mapsto \mathsf{next}(\pcconf)]$}}}
{\text{$(\lsconf,\timest,\logconf,\lkconf)\mpitrans{\loadact(\apr,\atr,\anaddr,\aval)} (\lsconf[\apr\mapsto s],\timest,\logconf,\lkconf)$}}\\[8mm]

\dfrac
{\splitfrac{\text{$\thestore{\anaddr}{\aval}\in\instrOf(\lsconf(\apr).\pcconf)$\quad
$\mathit{log} = \lsconf(\apr).\txnwrsconf[\anaddr\mapsto 1]$\quad
$\mathit{store} = \lsconf(\apr).\storeconf[\anaddr \mapsto\aval]$}}
{\text{$s = \lsconf(\apr)[\txnwrsconf \mapsto \mathit{log},\storeconf \mapsto \mathit{store},\pcconf\mapsto \mathsf{next}(\pcconf)]$}}}
{\text{$(\lsconf,\timest,\logconf,\lkconf)\mpitrans{\issueact(\apr,\atr,\anaddr,\aval)} (\lsconf[\apr\mapsto s],\timest,\logconf,\lkconf)$}}\\[8mm]

\dfrac
{\splitfrac{\splitfrac{\text{$\text{\plog{end}}\in\instrOf(\lsconf(\apr).\pcconf)$\quad $\forall\anaddr\in\mathbb{V}.\ \mathit{log}[\anaddr]=\bot \lor \timest(\anaddr) < \atr.\asti$ \quad $\mathsf{img}(\lsconf.\timest) < \atr.\acti$}}{\text{
$\mathit{Log} = \logconf[\anaddr\mapsto \storeconf[\anaddr]: \anaddr\in\mathbb{V}, \mathit{log}[\anaddr]\neq\bot]$  }}}
{\text{$\mathit{tstamp} = \timest[\anaddr\mapsto \asti.\acti: \anaddr\in\mathbb{V}, \mathit{log}[\anaddr]\neq\bot]$ \quad $s = \lsconf(\apr)[\pcconf\mapsto \mathsf{next}(\pcconf)]$}}}
{\text{$(\lsconf,\timest,\logconf,\lkconf)\mpitrans{\commitact(\apr,\atr)} (\lsconf[p\mapsto s],\logconf \mapsto \mathit{Log},\timest \mapsto \mathit{tstamp},\lkconf)$}}
}
\vspace{-3mm}
\caption{The set of transition rules defining snapshot isolation semantics model. We assume that all the events which come from the same transaction use a unique transaction identifier $\atr:(\asti,\acti)$ that has two components. For a function $f$, we use $f[a\mapsto b]$ to denote a function $g$ such that $g(c)=f(c)$ for all $c\neq a$ and $g(a)=b$. The function $\instrOf$ returns the set of instructions labeled by some given label while $\mathsf{next}$ gives the next instruction to execute.}
\label{Table:MPRules}
\end{figure}

Then, the semantics of a program $\aprog$ under snapshot isolation consistency model is defined using a labeled transition system (LTS) $[\aprog]_{\sic{}}=(\gstatesconf,\eventsconf,\gsconf_0,\rightarrow)$ where $\gstatesconf$ is the set of program configurations, $\eventsconf$ is a set of transition labels called \emph{events}, $\gsconf_0$ is the initial program configuration, and $\rightarrow\subseteq \gstatesconf\times \eventsconf\times \gstatesconf$ is the transition relation. 
The set of events under \sic{} is defined as follow.
\begin{align*}
\eventsconf =\ & \{\ \beginact(\apr,\atr), \loadact(\apr,\atr,\anaddr,\aval), \issueact(\apr,\atr,\anaddr,\aval), \commitact(\apr,\atr): \apr\in \mathbb{P}, \atr\in \mathbb{T}\times\mathbb{T}, \anaddr\in \mathbb{V}, \aval\in \mathbb{D}\}
\end{align*}
where $\beginact$ and $\commitact$ label transitions corresponding to the start and the commit of a transaction, respectively. 
$\issueact$ and $\loadact$ label transitions corresponding to writing, resp., reading, a shared variable during some transaction. 

The transition relation $\rightarrow$ is defined in Figure~\ref{Table:MPRules}. For readability, the events labeling a transition are written on top of $\rightarrow$. A $\beginact$ transition resets the local valuation of the shared variables and fetches their values from the central memory. A $\commitact$ transition applies the writes performed in a transaction to the central memory by merging the contents of the local copy $\storeconf$ with the central memory $\logconf$. An $\loadact$ transition reads the value of a shared-variable from the local copy $\storeconf$ while an $\issueact$ transition applies a new write to the local copy $\storeconf$. 

An execution of program $\aprog$, under snapshot isolation, is a sequence of events $\event_1\cdot\event_2\cdot\ldots$ labeling the transitions, such that there exists a sequence of configurations $\gsconf_0\cdot\gsconf_1\cdot\ldots$  where $\gsconf_0$ is the initial configuration before $\aprog$ starts execution and $\gsconf_{i-1}\xrightarrow{\event_{i}}\gsconf_{i}$  is a valid transition for $i > 1$.
 of transitions. The set of executions of $\aprog$ under \sic{} is denoted by $\executionsconf_{\sic{}}(\aprog)$.
\section{Proofs of Section~\ref{sec:Instr}: Characterizations of \sic{} Trace-Robustness}

In this section, we describe the proof of Theorem \ref{th:SIfinal}.
We first reduce the robustness of a program $\aprog$ against \sic{} to the existence of some execution trace $\atrace \in \tracesconf_{\sic{}}(\aprog) \setminus \tracesconf_{\serc{}}(\aprog)$ that has a specific shape. We call a  trace $\atrace \in \tracesconf_{\sic{}}(\aprog) \setminus \tracesconf_{\serc{}}(\aprog)$, an anomaly.
Then, we show that of the anomaly of a particular shape is equivalent to an execution of the instrumented program reaching an error state.

First, we give an auxiliary lemma about the happens-before relation (between events).
In the remaining of this section, we use $\hbo^{1}$ to denote the happens-before without the transitive closure, i.e., $\hbo^{1} = (\po \cup \sto \cup \rfo \cup \cfo \cup \sametro)$.
To decide if two events in a trace are ``independent'' (or commutative) we use the information about the existence of a happens-before relation between the events. If two events are not related by  happens-before then they can be swapped while preserving the same happens-before. Thus, we extend the happens-before relation to obtain the \emph{happens-before through} relation as follows:

\begin{definition}[\cite{DBLP:conf/esop/BouajjaniDM13}]\label{Definition:HBThrough}
Let $\tau = \alpha\cdot a\cdot \beta\cdot b\cdot \gamma$ be a trace where $a$ and $b$ are events (or atomic macro events), and $\alpha$, $\beta$, and $\gamma$ are sequences of events (or  atomic macro events) under a semantics $\sic{}$.
We say that \emph{$a$ happens-before $b$ through $\beta$} if there is a non empty sub-sequence $c_1 \cdots c_n$ of $\beta$ that satisfies:
$$c_i\rightarrow_{\hbo^{1}} c_{i+1}\quad\text{ for all }i\in[0, n]$$
where $c_0= a$, $c_{n+1}= b$.
\end{definition}

We deduce from the definition of serializable traces, that a anomaly trace must contain at least an issue and a commit events of the same transaction that are related via the happens-before through relation. Otherwise, we can build another trace with the same happens-before where events are reordered such that every issue $\issueact(\apr,\atr)$ is immediately followed by the corresponding commit $\commitact(\apr,\atr)$. The latter is a serializable trace which contradicts the initial assumption.

\begin{lemma}\label{lemma:ViolForm1}
Given an anomaly $\tau$, there must exist a transaction $\atr$ such that $\tau = \alpha \cdot \issueact(\apr,\atr) \cdot \beta \cdot \commitact(\apr,\atr) \cdot \gamma$ and $\issueact(\apr,\atr)$ happens before $\commitact(\apr,\atr)$ through $\beta$.
\end{lemma}

Given an anomaly of the from $\tau = \alpha \cdot \issueact(\apr,\atr) \cdot \beta \cdot \commitact(\apr,\atr) \cdot \gamma$, we call $\atr$ a \emph{delayed} transaction in the trace $\tau$ when $\issueact(\apr,\atr)$ happens before $\commitact(\apr,\atr)$ through $\beta$.

For an anomaly $\tau$, the \emph{number of delays}, denoted by $\#(\tau)$, in $\tau$ is the total number of \emph{delayed} transactions in the trace.
$$\mathtt{\#(\tau) = \#_{\atr\ is\ a\ \emph{delayed}\ transaction\ in\ \tau}\ \atr}$$

\begin{definition}[Minimal anomaly]
An anomaly $\tau$ is called \emph{minimal} if it has the least number of delays among all possible anomalies (for a given program $\aprog$).
\end{definition}

Given an anomaly $\tau = \alpha \cdot \issueact(\apr,\atr) \cdot \beta \cdot \commitact(\apr,\atr) \cdot \gamma$,  and assuming that $\atr$ is the \emph{first} delayed transaction in $\tau$ (w.r.t. the order between issue events of delayed transactions) and that $\tau$ is a minimal anomaly, the following lemma shows that we can assume w.l.o.g. that $\gamma$ is empty.

\begin{lemma}\label{lemma:ViolForm2}
Let $\tau = \alpha \cdot \issueact(\apr,\atr) \cdot \beta \cdot \commitact(\apr,\atr) \cdot \gamma$ be a minimal anomaly such that $\issueact(\apr,\atr)$ happens-before $\commitact(\apr,\atr)$ through $\beta$. Then, $\tau' = \alpha \cdot \issueact(\apr,\atr) \cdot \beta \cdot \commitact(\apr,\atr)$ is also a minimal anomaly.
\end{lemma}
\begin{proof}
We can notice that after executing the event $\commitact(\apr,\atr)$, we obtain a cycle in the $\hbo_t$ relation. Thus, $\tau'$ is already an anomaly not serializable.
\end{proof}

The following result relates \sic{} robustness problem to finding certain anomaly trace which is a minimal anomaly where only a single transaction is delayed.

\begin{theorem} \label{theorem:SIMinViol}
A program $\aprog$ is not robust under \sic{} iff there exists an anomaly $\tau$ under \sic{} such that the following must hold:
 $$\tau=\alpha\cdot \issueact(\apr,\atr) \cdot\beta\cdot \commitact(\apr,\atr)\mbox{ where:}$$
\begin{enumerate}[label=(\alph*)]
\item $\issueact(\apr,\atr)$ is the issue of the only delayed transaction in $\tau$; (Lemmas \ref{lemma:SIMinForm} and \ref{lemma:SIMinViol});
\item $\issueact(\apr,\atr)$ happens before $\commitact(\apr,\atr)$ through $\beta$ (Lemma \ref{lemma:SIMinForm});
\item for any event $a\in\beta$, we have that $(\issueact(\apr,\atr),a) \in \hbo$ and $(a,\commitact(\apr,\atr)) \in \hbo$ (Lemma \ref{lemma:SIMinForm});
\item there exist events $a$ and $b$ in $\beta$ such that $(\issueact(\apr,\atr),a) \in \cfo(x)$ and $(b,\commitact(\apr,\atr)) \in  \cfo(y)$ with $x \neq y$ (Lemma \ref{lemma:SIMinForm});
\item  all delayed transactions in $\beta$ don't write to shared variables that $\atr$ writes to (Lemma \ref{lemma:ViolForm-1}).
\end{enumerate}
\end{theorem}

Note that in certain cases the events $a$ and $b$ can be identical and $\beta = a$.
Figure~\ref{fig:sibadtrace} shows two examples of anomalies of the form given in Theorem \ref{theorem:SIMinViol}.
\begin{figure}[t]
\begin{minipage}[c]{0.43\textwidth}
\begin{subfigure}{\linewidth}
\scalebox{0.59}
{
\begin{tikzpicture}

 \node[shape=rectangle ,draw=none,font=\large] (A) at (0,0)  [] {$\issueact(\apr 1,\atr 1)$};
  \node[shape=rectangle ,draw=none,font=\large] (B) at (2.5,0)  [] {$(\apr 2,\atr 2)$};
  \node[shape=rectangle ,draw=none,font=\large] (C) at (5,0)  [] {$\storeact(\apr 1,\atr 1)$};

  \begin{scope}[ every edge/.style={draw=red,very thick}]
  \path [->] (A) edge [bend right] node [above,font=\large] {$RW$} (B);
  \path [->] (B) edge [bend right] node [above,font=\large] {$RW$} (C);
  \end{scope}

\end{tikzpicture}}
\caption{Anomaly trace of WS program.}
\label{fig:rob0tracesi}
\end{subfigure}
\end{minipage}
\hfill
\begin{minipage}[c]{0.47\textwidth}
\begin{subfigure}{\linewidth}
\scalebox{0.59}
{\begin{tikzpicture}

 \node[shape=rectangle ,draw=none,font=\large]  (A) at (0,0)  [] {$\issueact(\apr 3,\atr 3)$};
  \node[shape=rectangle ,draw=none,font=\large] (B) at (2.5,0)  [] {$(\apr 1,\atr 1)$};
 \node[shape=rectangle ,draw=none,font=\large]  (C) at (5,0)  [] {$(\apr 2,\atr 2)$};
  \node[shape=rectangle ,draw=none,font=\large] (D) at (7.5,0)  [] {$\storeact(\apr 3,\atr 3)$};

  \begin{scope}[ every edge/.style={draw=red,very thick}]
  \path [->,] (A) edge [bend left] node [above,font=\large] {$RW$} (B);
  \path [->] (B) edge  [bend left] node [above,font=\large] {$WR$} (C);
  \path [->] (C) edge  [bend left] node [above,font=\large] {$RW$} (D);
  \end{scope}
\end{tikzpicture}}
\caption{Anomaly trace of RWC program.}
\label{fig:rob1minitrace}
\end{subfigure}
\end{minipage}
\vspace{-2mm}
\caption{(a) Corresponds to an anomaly pattern where $\beta=(\apr 2,\atr 2)$, $\atr$ corresponds to $\atr 1$, and $a=b=(\apr 2,\atr 2)$. Also (b) Corresponds to an anomaly pattern where  $\beta=(\apr 1,\atr 1) \cdot (\apr 2,\atr 2)$, $\atr$ corresponds to $\atr 1$, and $a$ and $b$ correspond to $(\apr 1,\atr 1)$ and $(\apr 2,\atr 2)$, respectively.}
\label{fig:sibadtrace}
\end{figure}
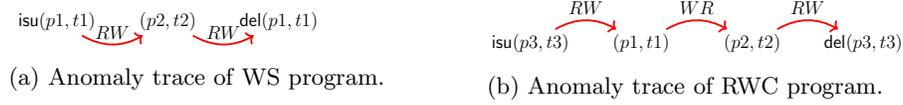

In the following, we give the lemmas that constitute Theorem \ref{theorem:SIMinViol}.


An important property in \sic{} semantics is that of conflict-free, and since the event $\commitact(\apr, \atr)$ is successfully executed only if there are no concurrent writes that were committed after $\issueact(\apr, \atr)$. Thus, for every event $\commitact(\apr_0,\atr_0)$ in $\beta$, $\commitact(\apr_0,\atr_0)$ don't write to a shared variable that $\commitact(\apr,\atr)$ writes to.

\begin{lemma}\label{lemma:ViolForm-1}
Let $\tau = \alpha \cdot \issueact(\apr, \atr)\cdot \beta\cdot \commitact(\apr, \atr)$ be a minimal anomaly such that $\issueact(\apr,\atr)$ happens-before $\commitact(\apr,\atr)$ through $\beta$. Then, for every  $a \in \beta$, $a$ does not write to a shared variable that $\commitact(\apr,\atr)$ writes to.
\end{lemma}

The following lemma shows that we must have an event $\issueact(\apr',\atr') \in \beta$ such that $\issueact(\apr',\atr')$ happens before $\commitact(\apr, \atr)$.

\begin{lemma}\label{lemma:ViolForm0}
Let $\tau = \alpha \cdot \issueact(\apr, \atr)\cdot \beta\cdot \commitact(\apr, \atr)$ be a minimal anomaly such that $\issueact(\apr,\atr)$ happens-before $\commitact(\apr,\atr)$ through $\beta$. Then, there must exist an event $\issueact(\apr',\atr') \in \beta$ such that $\issueact(\apr',\atr')$ happens before $\commitact(\apr, \atr)$, i.e., $(\issueact(\apr',\atr'),\commitact(\apr, \atr))\in \hbo$.
\end{lemma}

\begin{proof}
Suppose by contradiction that $\beta$ does not contain an event $\issueact(\apr',\atr')$ such that $\issueact(\apr',\atr')$ happens before $\commitact(\apr, \atr)$. Then, $\commitact(\apr,\atr)$ can be swapped with every $\issueact(\apr',\atr')$ event in $\beta$. Thus, we obtain that the two events $\issueact(\apr,\atr)$ and $\commitact(\apr,\atr)$ are adjacent which means that $\atr$ is no longer a delayed  transaction which is a contradiction. Therefore, $\beta$ does contain an event $\issueact(\apr',\atr')$ such that  $(\issueact(\apr',\atr'),\commitact(\apr, \atr))\in \hbo$.
\end{proof}

Next lemma shows that we can always obtain a minimal anomaly trace $\tau=\alpha\cdot \issueact(\apr,\atr) \cdot\beta\cdot \commitact(\apr,\atr)$ where $\beta$ contains no delayed transaction.
We show that if it were to have a delayed transaction $\atr_0$ in $\beta$, then it is possible to obtain a new anomaly where either $\atr$ is not delayed or $\atr_0$ is not delayed, and obtain a new anomaly with a smaller number of delayed transactions which contradicts the minimality assumption.

\begin{lemma}\label{lemma:SIMinViol}
Let $\tau=\alpha\cdot \issueact(\apr,\atr) \cdot\beta\cdot \commitact(\apr,\atr)$ be a minimal anomaly. Then, $\beta$ does not contain delayed transactions.
\end{lemma}

\begin{proof}
We suppose by contradiction that $\beta$ contains a delayed transaction $\atr_0$ issued by a process $\apr_0$.

It is important to notice that there must exist $\beta'\subset \beta$ and $\commitact(\apr_0,\atr_0) \in \beta$ such that $\issueact(\apr_0,\atr_0)$ happens before $\commitact(\apr_0,\atr_0)$ through $\beta$. Otherwise, we can commute the events until $\commitact(\apr_0,\atr_0)$ occurs just after $\issueact(\apr_0,\atr_0)$ and in this case
transaction $\atr_0$ is not delayed by $\apr_0$. Thus, $\tau$ is of the form
$\tau=\alpha\cdot \issueact(\apr,\atr) \cdot\beta_1\cdot\issueact(\apr_0,\atr_0)\cdot\beta'\cdot\commitact(\apr_0,\atr_0)\cdot\beta_2\cdot \commitact(\apr,\atr)$.

Notice that we can build happens-before cycle from $\issueact(\apr_0,\atr_0)$ to $\commitact(\apr_0,\atr_0)$. Also, since in $\beta_1\cdot\issueact(\apr_0,\atr_0)\cdot\beta'\cdot\commitact(\apr_0,\atr_0)\cdot\beta_2$ no transaction depends on $\issueact(\apr,\atr)$.
Thus, we can safely remove $\issueact(\apr,\atr)$ and its associated commit event $\commitact(\apr,\atr)$ and obtains:
$\tau'=\alpha\cdot\beta_1\cdot\issueact(\apr_0,\atr_0)\cdot\beta'\cdot\commitact(\apr_0,\atr_0)\cdot\beta_2$.
which is an anomaly because of the cycle formed as $\issueact(\apr_0,\atr_0)$ happens-before $\commitact(\apr_0,\atr_0)$ through $\beta'$.
In $\tau'$, transaction $\atr$ was not delayed, therefore, $\tau'$ has less number of delays than $\tau$. Thus, $\tau$ is not a minimal anomaly, a contradiction to our hypothesis.

\end{proof}

The next lemma characterizes the relation between the first delayed transaction and the commit of the underlying transaction. It shows the type of the first and last happens-before relations in the happens-before path between the issue of the only delayed transaction and its corresponding commit.

\begin{lemma}\label{lemma:SIMinForm}
Let $\tau = \alpha  \cdot \issueact(\apr,\atr) \cdot \beta  \cdot \commitact(\apr,\atr)$ be a minimal anomaly under \sic{}. Then, the following must hold:
There exist $(\apr_0,\atr_0),\ (\apr_1,\atr_1)\in\beta$ where
$(\issueact(\apr,\atr),(\apr_0,\atr_0)) \in \cfo$, and $((\apr_1,\atr_1),\commitact(\apr,\atr)) \in \cfo$.
\end{lemma}

\begin{proof}
We have that since $\issueact(\apr,\atr)$ happens-before $\commitact(\apr,\atr)$ through $\beta$ and $\beta$ does not contain a delayed transaction.
Then, there must exist $(\apr_0,\atr_0)$ and $(\apr_1,\atr_1)$ in $\beta$ such that $(\issueact(\apr,\atr),(\apr_0,\atr_0)) \in \hbo^{1}$, $((\apr_1,\atr_1),\commitact(\apr,\atr)) \in \hbo^{1}$, and $((\apr_0,\atr_0),(\apr_1,\atr_1)) \in \hbo$. It is important to note that $(\apr_0,\atr_0)$ and $(\apr_1,\atr_1)$ might be identical in which case $\beta = (\apr_0,\atr_0)$.

Notice that every event in $\beta$ (including $(\apr_0,\atr_0)$ and $(\apr_1,\atr_1)$) cannot write to a variable that $(\apr,\atr)$ writes to under \sic{} semantics, thus store order relation is not possible.
Also, since $(\apr,\atr)$ is not visible to any event in $\beta$ thus the read-from and program order are not possible. Thus, the only possibility is that $(\issueact(\apr,\atr),(\apr_0,\atr_0)) \in \cfo$, and $((\apr_1,\atr_1),\commitact(\apr,\atr)) \in \cfo$.

\end{proof}

\section{Proofs of Section~\ref{sec:Instr}: The Complete Instrumentation}\label{sec:InstrAppendix}
In this section, we present the instrumentation for the remaining instructions which are, begin and commit for the attacker and happens-before helper.

\subsection{Instrumentation of the Attacker}\label{subsec:Attacker}
We provide in Figure \ref{Figure:AttackerAppendix}, the instrumentation of the code for the attacker process. When the attacker randomly chooses a transaction to delay, it sets the flag $a_{\attacktransaction}$ to $\mytrue$ in the instruction $\beginact$ (line~\eqref{Equation:AttackerTransaction}). Then, it sets the flag $p.a$ to $\ajoin$ to indicate that the current process is the attacker. It copies the values from every variable $x$ to its primed version $x'$.

In the case the attacker starts the happens-before chain, it has to set the variable $\hb{}$ to $\mytrue$ to mark the start of the happens-before chain and the end of the visibility chain and set the flag $x.event$ to $\loadacc$ (line~\eqref{Equation:AttackerLastLoad} in Figure \ref{Figure:Attacker}). We can notice that when the $\hb{}$ is set to $\mytrue$, we can no longer execute new transactions from the attacker (all conditions in lines~\eqref{Equation:AttackerBegin} and \eqref{Equation:AttackerBeginDelay} become $\myfalse$).

\begin{figure}[htbp]
 \scriptsize
 \begin{minipage}{0.49\linewidth}
\begin{eqnarray}
&&\semidler{\thetransition{\alab_1}{\alab_2}{\commitact}}=\notag\\
&&\textbf{// Typical execution of commit}\notag\\
&&\thetransitionInstrumented{\alab_1}{\alab_{x1}}{\thecondition{a_{\attacktransaction} = \perp \lor \ p.a \neq \perp}}\notag\\
&&\thetransitionInstrumented{\alab_{x1}}{\alab_{2}}{\commitact}\notag\\
&&\textbf{// Succeeded Commit of delayed transaction}\notag\\
&&\thetransitionInstrumented{\alab_1}{\alab_{x2}}{\thecondition{ a_{\attacktransaction} \neq \perp \land \ \hb{} \neq \perp}}\notag\\
&&\thetransitionInstrumented{\alab_{x2}}{\alab_2}{\commitact}\notag\\
&&\textbf{// Failed commit of delayed transaction}\notag\\
&&\thetransitionInstrumentedb{\alab_1}{\thecondition{a_{\attacktransaction} \neq \perp \land \ \hb{} = \perp}}\label{Equation:ABLOCK}
\end{eqnarray}
 \end{minipage}\hfill
 \begin{minipage}{0.49\linewidth}
 \begin{eqnarray}
&&\semidler{\thetransition{\alab_1}{\alab_2}{\beginact}} = \notag\\
&&\textbf{// Typical execution of begin}\notag\\
&&\thetransitionInstrumented{\alab_1}{\alab_{x1}}{\thecondition{\hb{} = \perp \land (p.a \neq \perp \lor  a_{\attacktransaction} = \perp)}}\ \ \ \ \ \  \label{Equation:AttackerBegin}\\
&&\thetransitionInstrumented{\alab_{x1}}{\alab_2}{\beginact}\notag\\
&&\textbf{// Begin of delayed transaction}\notag\\
&&\thetransitionInstrumented{\alab_1}{\alab_{x2}}{\thecondition{\hb{} = \perp \land \  a_{\attacktransaction} = \perp}}\label{Equation:AttackerBeginDelay}\\
&&\thetransitionInstrumented{\alab_{x2}}{\alab_{x3}}{\beginact}\notag\\
&&\thetransitionInstrumented{\alab_{x3}}{\alab_{x4}}{\theassign{p.a}{\ajoin}}\notag\\
&&\thetransitionInstrumented{\alab_{x4}}{\alab_{x5}}{\theforeachassign{\anaddr}{\mathbb{V}}{\anaddr'}{\anaddr}}\notag\\
&&\thetransitionInstrumented{\alab_{x5}}{\alab_2}{\theassign{a_{\attacktransaction}}{\mytrue}}\label{Equation:AttackerTransaction}
\end{eqnarray}
 \end{minipage}
\normalsize
\caption{Instrumentation of the Attacker}
\label{Figure:AttackerAppendix}
\end{figure}

\subsection{Instrumentation of the Happens-Before Helpers}\label{subsec:HBHelpers}
In Figure \ref{Figure:HBHelpersInstrumentationAppendix}, we provide the instrumentation of the remaining instructions of a happens-before helper.

When the flag $\hbo$ is set to $\mytrue$, a process (which cannot be the attacker, i.e., the flag $p.a$ is null) starts the attempts to join the set of happens-before helpers. Thus, it randomly chooses a first transaction (the $\beginact$ of this transaction is shown in line~\eqref{Equation:HBHTransaction}) through which the process will join the set of happens-before helpers. When a process chose the transaction to join happens-before helpers, that means it has made pledge that during this transaction it will either do read from a variable that was updated by a another delayed transaction from some other process in happens-before helpers or write to a variable that was accessed with a read or write from another process in happens-before helpers. When either one of these criteria are satisfied the flag $p.hbh$ will be set to $\hbhjoin$. If a process does not keep its pledge (the flag $p.hbh$ is null) then before executing the $\commitact$ instruction of the first transaction we block the execution (line~\eqref{Equation:HBHBLOCK}).

The happens-before helpers processes continue executing their instructions, until one of them executes a load that reads from the shared variable $\anaddr$ that was stored in $a_{\attackstore}$ which implies the existence of a happens-before cycle.
When executing the instruction $\commitact$ at the end of every transaction, we have a conditional check to detect if we have a load or a write  accessing the variable $\anaddr$ (lines~\eqref{Equation:HBHFINISH1}, \eqref{Equation:HBHFINISH2}, and \eqref{Equation:HBHFINISH3}).
When the check detects that the variable $\anaddr$ was accessed, the execution goes to the error state (line~\eqref{Equation:HBHFINISH3}) to indicate that the execution has produced an anomaly and we denote the reached state of the instrumented program's execution, the error state.

\begin{figure}[htbp]
 \scriptsize
\begin{minipage}{0.49\linewidth}
\begin{eqnarray}
&&\semHbhelperorig{\thetransition{\alab_1}{\alab_2}{\beginact}}=\notag\\
&&\textbf{// Begin before joining happens-before}\notag\\
&&\thetransitionInstrumented{\alab_1}{\alab_{x1}}{\thecondition{\hb{} = \perp \land p.a = \perp}}\notag\\
&&\thetransitionInstrumented{\alab_{x1}}{\alab_2}{\beginact}\notag\\
&&\textbf{// Begin of first transactions to join}\notag\\
&&\textbf{// happens-before}\notag\\
&&\thetransitionInstrumented{\alab_1}{\alab_{x2}}{\thecondition{\hb{} \neq \perp \land   p.hbh =  p.a = \perp}}\notag\\
&&\thetransitionInstrumented{\alab_{x2}}{\alab_{x3}}{\beginact}\label{Equation:HBHTransaction}\\
&&\textbf{// Copying flags before the first transaction}\notag\\
&&\textbf{// starts executing and modifying the flags}\notag\\
&&\thetransitionInstrumented{\alab_{x3}}{\alab_2}{\theforeachassign{\anaddr}{\mathbb{V}}{\anaddr.eventI}{\anaddr.event}}\notag\\
&&\textbf{// Begin after joining happens-before}\notag\\
&&\thetransitionInstrumented{\alab_1}{\alab_{x4}}{\thecondition{\hb{} \neq \perp \land \   p.hbh \neq \perp}}\notag\\
&&\thetransitionInstrumented{\alab_{x4}}{\alab_2}{\beginact}\notag
\end{eqnarray}
 \end{minipage}\hfill
\begin{minipage}{0.49\linewidth}
\begin{eqnarray}
&&\semHbhelperorig{\thetransition{\alab_1}{\alab_2}{\commitact}}=\notag\\
&&\textbf{// Commit before joining happens-before}\notag\\
&&\thetransitionInstrumented{\alab_1}{\alab_{x1}}{\thecondition{\hb{}= \perp \land p.a = \perp}}\notag\\
&&\thetransitionInstrumented{\alab_{x1}}{\alab_2}{\commitact}\notag\\
&&\textbf{// Commit after joining happens-before}\notag\\
&&\thetransitionInstrumented{\alab_1}{\alab_{x2}}{\thecondition{\hb{} \neq \perp \land \  p.hbh \neq \perp}}\notag\\
&&\thetransitionInstrumented{\alab_{x2}}{\alab_{x3}}{\commitact}\notag\\
&&\thetransitionInstrumented{\alab_{x3}}{\alab_{x4}}{\theload{\tilde r}{a_{\attackstore}}}\label{Equation:HBHFINISH1}\\
&&\thetransitionInstrumented{\alab_{x4}}{\alab_{x5}}{\theassign{\tilde r}{\tilde r .event}}\label{Equation:HBHFINISH2}\\
&&\textbf{// Commit after reaching the error state}\notag\\
&&\thetransitionInstrumentedf{\alab_{x5}}{\thecondition{\tilde r \neq \perp}}\label{Equation:HBHFINISH3}\\
&&\thetransitionInstrumented{\alab_{x5}}{\alab_{2}}{\thecondition{\tilde r = \perp}}\notag\\
&&\textbf{// Failed attempt to join happens-before}\notag\\
&&\thetransitionInstrumentedb{\alab_1}{\thecondition{\hb{} \neq \perp \land \  p.hbh = p.a = \perp}}\ \ \ \ \ \ \label{Equation:HBHBLOCK}
\end{eqnarray}
 \end{minipage}
\normalsize
\caption{Instrumentation of Happens-Before Helpers Processes}
\label{Figure:HBHelpersInstrumentationAppendix}
\end{figure}

As a direct consequence of Theorem~\ref{th:SIfinal}, the next corollary states that some programs which have certain characteristics are robust against \sic{}.
\begin{corollary}\label{lemma:SIRobustExe}
Given a program $\aprog$, if one of the following holds:
\begin{enumerate}[label=(\alph*)]
\item  every transaction of $\aprog$ contains a single instruction either a read or a write;
\item  every transaction of $\aprog$ contains only read/write events that access a single variable (different transactions might read/write to different variables);
\item  given a variable $x$, every transaction of $\aprog$ contains a write to the variable $x$.
\end{enumerate}
then $\aprog$ is robust under \sic{}.
\end{corollary}

\section{Proofs of Section~\ref{sec:Instr}: Soundness and Completeness of the Instrumentation} \label{subsec:SCInstr}

The aim of the instrumentation procedure is to reduce the problem of checking the existence of the anomaly described in Theorem \ref{theorem:SIMinViol} to reachability under serialisability of en error state by the instrumented version of a program. The instrumentation procedure is considered sound and complete iff if an error state is reachable, then we can reconstruct an anomaly, and every anomaly ensures that the error state is reachable by the instrumented version of the program.

\begin{theorem}[Soundness and Completeness]
A program $\aprog$ is not robust iff the instrumented version of it, $\aprog'$, reaches an error state under \serc{}.
\end{theorem}

\begin{proof}

\textbf{Soundness.}
Suppose that the instrumented program reaches an error state.
Then, the execution's trace of the instrumented program is of the form:
$$\tau_\errorstate = \tau_1\cdot \issueact(\apr,\atr) \cdot\tau_2\cdot (\apr',\atr') $$
The last transaction, $(\apr',\atr')$ performed by a process $\apr''$ that has a read accessing the variable $x = a_{\attackstore}$ and is part of the happens-before helpers.
Because the conditional check can be performed only by a process ($\apr_{HbH1}$) that is one of the happens-before helpers and is currently executing.

In order for $\apr_{HbH1}$, to join the set of happens-before helpers, it must have found that the valuation of the flag $\hb{}$ is not null which means there exists some process $\apr$ that is the attacker that sets the flag $\hb{}$ to $\mytrue$. In $\tau_1$, the attacker, happens-before helpers, and other processes start executing the original instructions without setting any flags or delaying any transactions.
Afterwards, the attacker issues the delayed transaction $\issueact(\apr,\atr)$ and it starts populating the primed variables $x'$ and reading from them and setting the flags $x.event'$ to $1$ for every variable $x$ that it writes to and $y.event$ to $\loadacc$ for every variable $y$ that it reads from. During the execution of $\atr$, the attacker sets the flag $\hb{}$ to $\mytrue$. Hence, the happens-before helpers start checking at every instruction whether the flags $x.event$ are set to either $\storeacc$ or $\loadacc$. If so, they start populating the flags $x.event$ and $\alock.event$  as well. When $\hb{}$ is set to $\mytrue$, the attacker stop issuing new transactions. Therefore, all transaction in $\tau_2$ are from the happens-before helpers.

We now transform $\tau_\errorstate$ into the following execution trace:
$$\tau = \tau'_1\cdot \issueact(\apr,\atr) \cdot\tau'_2\cdot \commitact(\apr,\atr)$$

Here, $\tau_1'$ is the subsequence of all $\tau_1$ events that are produced by instructions from $\aprog$ without the conditionals checking (i.e., the assume statements).
The transaction $\atr$ which is executed by the attacker represents the delayed transactions in $\tau$ with the removal of the conditionals checking and the flags setting.
$\tau_2'$ is the subsequence of all events of $\tau_2$ produced by transactions from $\aprog$ which are executed only by the happens-before helpers except the conditionals checking and the flags setting.
We add the commit of transaction $\commitact(\apr_0,\atr)$ to describe the commit of the delayed transaction that was delayed by the attacker. $\tau$ is a possible execution's trace of the program $\aprog$ because $\tau_\errorstate$ is result of an execution of the instrumented version of $\aprog$ and we have removed from $\tau$ all the effects of the instrumentation, and replaced the stores to auxiliary variables by issues of stores without changing the dependency between all the events in the execution.

All transactions in $\tau_2'$ are from the happens-before helpers. Transactions in $\tau_2'$ form a happens-before path between $\issueact(\apr,\atr)$ and $\commitact(\apr,\atr)$. Also, we have $a,\ b=(\apr',\atr') \in\tau_2'$ such that $(\issueact(\apr,\atr),a)\in \cfo(y)$ and $(b,\commitact(\apr,\atr)) \in \cfo(x)$. No transaction in $\tau_2'$ writes to a variable that $\atr$ writes to.
Hence, $\tau$ indeed holds all the properties of the anomaly described in Theorem \ref{theorem:SIMinViol}.\\[0.2cm]
\textbf{Completeness.}
Suppose we have an anomaly of a given program $\aprog$:
$$\tau = \tau_1\cdot \issueact(\apr,\atr) \cdot\tau_2\cdot \commitact(\apr,\atr)$$
such that $\tau$ maintains all the properties given in Theorem \ref{theorem:SIMinViol}. We demonstrate that there is a possible serializable execution based on $\tau$ of the instrumented version of the program $\aprog$ that reaches the error state. Next, we show how to build the instrumented program execution.

At the start of the execution, $\tau_1$, the attacker, happens-before helpers, and other processes execute the original transactions with just conditional checks.

Afterwards, the attacker delays the transaction $\issueact(\apr,\atr)$ and starts populating the flags. In $\issueact(\apr,\atr)$, the attacker issues a store to the shared variable $`x`=a_{\attackstore}$ and $\exists\ b\in\tau_2$ such that $(b,\commitact(\apr,\atr)) \in \cfo(x)$. All writes that were executed in $\atr$ by the attacker are invisible to the remaining processes which includes the happens-before helpers. While executing $\atr$, the attacker sets the content of the flag $y.event$ to $\loadacc$ for every variable $y$ that it reads from and it sets the flag $\hbo{}$ to $\mytrue$.

On the other hand, the processes which are executing their transactions without delaying them will attempt to join the happens-before helpers by checking if the flag $\hbo{}$  is set to $\mytrue$. If so, they start the attempt of joining the happens-before helpers and when it succeed they joining the happens-before helpers and  start executing their transactions which constitute $\tau_2$. The first executed transaction by the happens-before helpers is $a$ described above which signals the start of $\tau_2$ and the happen before dependency. Thus, in $\tau_2$, we have only transactions form the happens-before helpers (because the attacker stop when the flag $\hbo{}$ is set to $\mytrue$) such that they are related by the happen before dependency that started from $\issueact(\apr,\atr)$ until it reaches $\commitact(\apr,\atr)$ through $\tau_2$. We know that there must exist $b\in\tau_2$ such that $(b,\commitact(\apr,\atr))\in \cfo(`x`=a_{\attackstore})$. $b$ is equivalent to the last executed transaction by the happens-before helpers that accesses the shared variable $x$. Thus, the underlying happens-before helper will set the content of the flag $x.event$ to $\loadacc$. Hence, when the underlying process executes the $\commitact$ instruction of this transaction, it will go to the error state (lines \eqref{Equation:HBHFINISH1}, \eqref{Equation:HBHFINISH2}, and \eqref{Equation:HBHFINISH3}) and in this case the instrumented version of the program $\aprog$  has reached the desired error state.
\end{proof}

\section{Proofs of Section~\ref{sec:PPR}}

The following theorem shows that the acyclicity of the commutativity dependency graph of a program implies the robustness of the program. Actually, the notion of robustness in this theorem relies on a slightly different notion of trace where store-order and write-order dependencies take into account values, i.e., store-order relates only transactions writing different values and the write-order relates a reading transaction to the oldest transaction (w.r.t. execution order) writing its value. In more details, we assume that two transactions are $\sto$-related iff they write different values. Notice that since $\cfo$ is defined using $\sto$ then when two transactions are $\sto$-related iff they write different values then when two transactions are $\cfo$-related this implies that the value that is read is different than the one that is written. For example, in Figure \ref{fig:robsto}, if we don't have this weakening of $\sto$ we get an execution trace, where the happens-before is cyclic, of the form $\tau = \issueact(\apr_1,\atr_1) \cdot (\apr_2,\atr_2) \cdot (\apr_3,\atr_3) \cdot \commitact(\apr_1,\atr_1)$ because of the $\sto$ relation that links $(\apr_2,\atr_2)$ and $(\apr_3,\atr_3)$. However, with our weakening of $\sto$ there will be no $\sto$ relation between $(\apr_2,\atr_2)$ and $(\apr_3,\atr_3)$ and this the above execution can be equivalently rewritten as $\tau = (\apr_3,\atr_3) \cdot  \issueact(\apr_1,\atr_1) \commitact(\apr_1,\atr_1) \cdot (\apr_2,\atr_2)$ which is serializable.
Similarly, we assume that two transactions $\atr_1$ and $\atr_2$ are related by $\rfo$ iff when we swap the two transactions $\atr_2$ does not read the same value that $\atr_1$ is writing. For instance, in Figure \ref{fig:robrfo}, if we don't have this assumption we get an execution trace, where the happens-before is cyclic, of the form $\tau = \issueact(\apr_1,\atr_1) \cdot (\apr_2,\atr_2) \cdot (\apr_3,\atr_3) \cdot \commitact(\apr_1,\atr_1)$ because of the $\rfo$ relation that links $(\apr_2,\atr_2)$ and $(\apr_3,\atr_3)$. However, with our assumption on $\rfo$ there will be no $\rfo$ relation between $(\apr_2,\atr_2)$ and $(\apr_3,\atr_3)$ and this the above execution can be equivalently rewritten as $\tau = (\apr_3,\atr_3) \cdot  \issueact(\apr_1,\atr_1) \commitact(\apr_1,\atr_1) \cdot (\apr_2,\atr_2)$ which is serializable.
This approach helps in avoiding some of the harmless anomalies, where the happens before cycle might be caused by a write-write dependency between two transactions that writes the same values.
\begin{figure}[t]
\lstset{basicstyle=\ttfamily\scriptsize}
\begin{subfigure}{59mm}
\begin{minipage}[c]{19mm}
\begin{lstlisting}[language=Java10]
    p1:
t1: [r1 = y //0
     x = 1]
\end{lstlisting}
\end{minipage}
\begin{minipage}[c]{2mm}
\footnotesize{$||$}
\end{minipage}
\begin{minipage}[c]{15mm}
\begin{lstlisting}[language=Java10]
    p2:
t2: [y = 1
     z = 1]
\end{lstlisting}
\end{minipage}
\begin{minipage}[c]{2mm}
\footnotesize{$||$}
\end{minipage}
\begin{minipage}[c]{17mm}
\begin{lstlisting}[language=Java10]
    p3:
t3: [z = 1
     r3 = x] //0
\end{lstlisting}
\end{minipage}
\caption{Robust program when weakening $\sto$.}
\label{fig:robsto}
\end{subfigure}
\hspace{5mm}
\begin{subfigure}{59mm}
\begin{minipage}[c]{19mm}
\begin{lstlisting}[language=Java10]
    p1:
t1: [r1 = y //0
     x = 1]
\end{lstlisting}
\end{minipage}
\begin{minipage}[c]{2mm}
\footnotesize{$||$}
\end{minipage}
\begin{minipage}[c]{15mm}
\begin{lstlisting}[language=Java10]
    p2:
t2: [y = 1
     z = 0]
\end{lstlisting}
\end{minipage}
\begin{minipage}[c]{2mm}
\footnotesize{$||$}
\end{minipage}
\begin{minipage}[c]{17mm}
\begin{lstlisting}[language=Java10]
    p3:
t3: [r2 = z //0
     r3 = x] //0
\end{lstlisting}
\end{minipage}
\caption{Robust program when weakening $\rfo$.}
\label{fig:robrfo}
\end{subfigure}
\caption{Examples of robust programs.}
\end{figure}

\begin{theorem}
For a program $\aprog$, if (1) the commutativity dependency graph of $\aprog$ does not contain non-mover cycles, then (2) $\aprog$  is robust.
\end{theorem}
\begin{proof}
We prove the contrapositive, i.e., $\neg(2)\Rightarrow\neg(1)$. For the proof, we use the result of Theorem \ref{th:SIfinal}.

Assuming that the program $\aprog$ is not robust. Then, based on Theorem \ref{th:SIfinal} there must exist an execution of the instrumentation of $\aprog$ that reaches the error state. We suppose that $\atr$ is the delayed transaction, $\atr_{ins}$ is the instrumentation of $\atr$ (writes are stored in auxiliary registers), and $\apr$ is the attacker process.  Therefore, the  execution of the instrumentation of $\aprog$ that reaches the error state is of the form $\tau=\alpha\cdot (\apr,\atr_{ins}) \cdot a \cdot\beta\cdot b$ where $a$ writes to a variable that $\atr$ reads from and $b$ reads from a variable that $\atr$ writes to.
We assume that $b$ is the first event that does read that accesses a variable that $\atr$ writes to.
In the following we show that  the commutativity dependency graph of $\aprog$ contains a non-mover cycle where $\atr$ is $\atr_0$. We consider two cases, first case
when $a=b$ and $\beta = \epsilon$, and second case is when $a \neq b$.

First case: $\tau=\alpha\cdot (\apr,\atr_{ins}) \cdot a $ where $a$ writes to a variable that $\atr$ reads from, reads from a variable that $\atr$ writes to, and does not write to a variable that $\atr$ writes to. Assume that $a = (\apr_1,\atr_1)$.
Thus, we can safely obtain $\tau_0=\alpha\cdot  (\apr_1,\atr_1) \cdot (\apr,\atr')$ a serializable execution trace of $\aprog$ where $\atr'$ is the reads free instantiation of $\atr$. Since $((\apr_1,\atr_1),\commitact(\apr,\atr)) \in \cfo$ then $\atr_1$ reads a value that $\atr'$ is overwriting with a different value. Therefore, $\tau'_0=\alpha\cdot(\apr,\atr')\cdot(\apr_1,\atr_1)$  is either a serializable execution with a different end state than $\tau_0$ has or it is not an serializable execution. Thus, $(\atr_1,\atr') \in \mcfo$ and $\atr_1$ does not write to a variable that $\atr$ writes to. Similarly, we can safely obtain $\tau_n=\alpha\cdot  (\apr,\atr'') \cdot (\apr_1,\atr_1)$ a serializable execution trace of $\aprog$ where $\atr''$ is the writes free instantiation of $\atr$. Since $(\issueact(\apr,\atr),(\apr_1,\atr_1)) \in \cfo$ then $\atr''$ reads a value that $\atr_1$ is overwriting with a different value. Therefore, $\tau'_n=\alpha\cdot(\apr_1,\atr_1)\cdot(\apr,\atr'')$  is either a serializable execution with a different end state than $\tau_n$ has or it is not an serializable execution. Thus, $(\atr'',\atr_1) \in \mcfo$ and $\atr_1$ does not write to a variable that $\atr$ writes to.

Second case: $\tau=\alpha\cdot (\apr,\atr_{ins}) \cdot a \cdot \beta \cdot b $ where $a$ writes to a variable that $\atr$ reads from, $b$ reads from a variable that $\atr$ writes to, and every transaction in $a\cdot\beta\cdot b$ does not write to a variable that $\atr$ writes to. Assume that $a = (\apr_1,\atr_1)$ and $b= (\apr_n,\atr_n)$. Since the transactions in $(\apr_1,\atr_1) \cdot \beta \cdot (\apr_n,\atr_n)$ constitute the happens-before path in the execution trace $\tau$. Then, for every $(\apr_i,\atr_i),\ (\apr_{i+1},\atr_{i+1}) \in (\apr_1,\atr_1) \cdot \beta \cdot (\apr_n,\atr_n)$ we have that $((\apr_i,\atr_i),(\apr_{i+1},\atr_{i+1})) \in (\po \cup \rfo \cup \sto \cup \cfo)$. In the case $((\apr_i,\atr_i),(\apr_{i+1},\atr_{i+1})) \in (\rfo \cup \sto \cup \cfo)$, we can safely obtain $\tau_i=\alpha\cdot \gamma \cdot (\apr_i,\atr_i) \cdot (\apr_{i+1},\atr_{i+1})$ which is a serializable execution trace of $\aprog$ where $\gamma$ either empty (i.e., $\epsilon$) or $\gamma = (\apr_1,\atr_1) \cdot \cdots \cdot (\apr_{i-1},\atr_{i-1})$.  Since, $((\apr_i,\atr_i),(\apr_{i+1},\atr_{i+1})) \in (\rfo \cup \sto \cup \cfo)$, then swapping $\atr_i$ and $\atr_i+1$ will result in either reordering of writes or write overwrites a read, or read obtains a different value. Therefore, $\tau'_i=\alpha\cdot \gamma \cdot (\apr_{i+1},\atr_{i+1})\cdot (\apr_i,\atr_i) $  is either a serializable execution trace with a different end state than $\tau_i$ has or it is not an serializable execution trace. Thus, $(\atr_i,\atr_{i+1}) \in \mcfo$. Also, we have that $\atr_{i}$ and $\atr_{i+1}$ do not write to a variable that $\atr$ writes to. Similar to the first case, we can safely obtain $\tau_0=\alpha\cdot  (\apr_1,\atr_1)\cdot \beta \cdot (\apr_n,\atr_n) \cdot (\apr,\atr')$ a serializable execution trace of $\aprog$ where $\atr'$ is the reads free instantiation of $\atr$. Since $((\apr_n,\atr_n),\commitact(\apr,\atr)) \in \cfo$ then $\atr_n$ reads a value that $\atr'$ is overwriting with a different one. Therefore, $\tau'_0=\alpha\cdot  (\apr_1,\atr_1)\cdot \beta \cdot (\apr,\atr')\cdot (\apr_n,\atr_n) $  is either a serializable execution with a different end state than $\tau_0$ has or it is not an serializable execution trace. Thus, $(\atr_n,\atr') \in \mcfo$ and $\atr_n$ does not write to a variable that $\atr$ writes to. Furthermore, we can safely obtain $\tau_n=\alpha\cdot  (\apr,\atr'') \cdot (\apr_1,\atr_1)$ a serializable execution  trace of $\aprog$ where $\atr''$ is the writes free instantiation of $\atr$. Since $(\issueact(\apr,\atr),(\apr_1,\atr_1)) \in \cfo$ then $\atr''$ reads a value that $\atr_1$ is overwriting with a different one. Then, $\tau'_n=\alpha\cdot(\apr_1,\atr_1)\cdot(\apr,\atr'')$  is either a serializable execution trace with a different end state than $\tau_n$ has or it is not an serializable execution trace. Thus, $(\atr'',\atr_1) \in \mcfo$ and $\atr_1$ does not write to a variable that $\atr$ writes to.




\end{proof}

\section{Experiments Appendix}\label{sec:Exper}

In this section we describe the applications we used to evaluate our techniques.
For every table in the original application program, we added a boolean annotation in our formalization of the table in Boogie, in order,
to inspect whether a given record does exist in the table or not. For instance, if we consider the following table \emph{Customer(CustomerId, CustomerName)}, in Boogie, we formalize the table
as two maps: \emph{CustomerAlive} which takes a \emph{CustomerId} and returns \emph{true} if there is a customer with underlying id and \emph{else} otherwise,  \emph{CustomerTable} which takes a \emph{CustomerId} and returns the corresponding \emph{CustomerName}. Also, in certain cases we formalize a Table as multiple maps.
Below, we describe each application.

\smallskip
\noindent
{\bf Auction \cite{DBLP:conf/concur/0002G16}:}  Which has five transactions that manipulate three tables: BIDS, ITEMS, and USERS.
Transaction RegBid is for placing a bid on an item. Transaction RegUser is for registration a user. Transaction ViewItem is for viewing the number of bids for an item.
Transaction ViewUser is for looking at a user name. Transaction ViewUsers is for looking at all registered users.

\smallskip
\noindent
{\bf Cassieq-Core\footnote{https://github.com/paradoxical-io/cassieq}:} A core unit of a distributed queue. It has eight transactions that manipulate a single table: USERACCOUNTS.
Transaction AddNewAccount is for adding a new account in USERACCOUNTS.
Transaction DeleteAnAccount is for deleting an account from USERACCOUNTS.
Transaction AddNewKey is for adding a new key to an existing account in USERACCOUNTS.
Transaction DeleteAKey is removing a key from an existing account in USERACCOUNTS.
Transaction GetAnAccount is to check whether there exist an account with a given id in the table USERACCOUNTS.
Transaction GetAccounts is to return all existing accounts in USERACCOUNTS.
Transaction GetAccountKeys is for getting all the keys of certain account in USERACCOUNTS.
Transaction GetAccountKey is to check whether a certain account does hold a certain key in USERACCOUNTS.

\smallskip
\noindent
{\bf Courseware \cite{DBLP:conf/popl/GotsmanYFNS16,DBLP:conf/concur/NagarJ18}:} Which has five transactions that manipulate three tables: STUDENT, COURSE, and ENROLED.
Transaction RegisterStudent is for registering a new student in the table STUDENT.
Transaction AddCourse is for adding a new course in the table COURSE.
Transaction EnrollStudent is for enrolling a given registered student in a given course.
Transaction RemoveCourse is for removing a given course from the table COURSE.
Transaction QueryCourses is for inspecting courses in the table COURSE.

\smallskip
\noindent
{\bf Currency-Exchange\footnote{https://github.com/Haiyan2/Trade}:} A trading service. It has six transactions that manipulate a single table: TRADES.
Transaction SaveTrade is for registering a new trade.
Transaction ViewListTrades is for viewing the trades that occurred before a given instance of time.
Transaction ViewTrade is for inspecting a given trade.
Transaction ViewTradeUser is for looking for a user who carried out a given trade.
Transaction GetNbTrades is for inspecting the number of trades.
Transaction GetTradeTimeStamp is for looking for the time stamp of a given trade.

\smallskip
\noindent
{\bf FusionTicket \cite{DBLP:conf/cloud/HoltBZPOC16}:}  Which has four transactions that manipulate a single table: EVENTS.
Transaction AddEvent is adding new event in some given venue. Transaction ViewEvent is for looking at given a event and the number of tickets available at this event.
Transaction Browse is for looking at events that are planned in some given venue. Transaction Purchase is for buying a ticket at a certain event.

\smallskip
\noindent
{\bf Shopping-Cart\footnote{https://github.com/nikhilswagle/Shopping_Cart_Angular_Cassandra}:} An on-line shop service implemented over Cassandra. It has four transactions that query two tables: USERS and PRODUCTS.
Transaction GetUser is for querying the existence of a user in the table USERS.
Transaction GetProductsByCategory is for finding products that are in a given category.
Transaction GetProductByUPC is for finding a product through its UPC.
Transaction GetCategories is for finding the categories.

\smallskip
\noindent
{\bf Playlist\footnote{https://github.com/DataStaxDocs/playlist}:} An on-line music service.  It has fourteen transactions that manipulate three tables: USERS, TRACKS, and ARTISTS.
Transaction AddTrack is for adding a new track in the table TRACKS.
Transaction GetTrack is for getting a certain track from the table TRACKS.
Transaction AddUser is for adding a new user in the table USERS.
Transaction GetUser is for looking for a certain user in the table USERS.
Transaction CreatePlayList is for creating a playlist of certain user in the table USERS.
Transaction ListArtistByLetter is for listing artists by their first letters of their names.
Transaction ListSongsByArtist is for listing tracks produced by certain artist.
Transaction ListSongsByGenre is for listing tracks of certain genre type.
Transaction AddTrackToPlaylist is for adding an existing track (in TRACKS) to an existing user play list in the table USERS.
Transaction DeleteTrackFromPlaylist is for removing a track from user play list in the table USERS.
Transaction GetPlaylistForUser is for getting the contents of certain play list of certain user.
Transaction GetPlaylistNames is for getting all the play lists of certain user.
Transaction DeletePlayListForUser is for deleting a certain user's play list.
Transaction DeleteUser is deleting a user from the table USERS.

\smallskip
\noindent
{\bf RoomStore\footnote{https://github.com/mebigfatguy/roomstore}:} A messages bot service.  It has five transactions that manipulate a single table: MESSAGES.
Transaction AddMessage is for adding a new message to the table MESSAGES.
Transaction GetLastMessage is for getting the messages of given user.
Transaction GetMessages is for looking for messages that were added in a certain date.
Transaction GetSpecificMessage is for getting specific message that was added in a certain date and time.
Transaction GetTopicMessages is for getting messages that are of certain topic.

\smallskip
\noindent
{\bf SmallBank \cite{DBLP:conf/icde/AlomariCFR08}:} Which has five transactions that manipulate three tables: ACCOUNT, SAVING, and CHECKING.
Transaction Balance is for looking at both the saving and checking balances of a given user account.
Transaction DepositChecking is for depositing a certain amount into the checking balance.
Transaction TransactSaving is for depositing or withdrawing into/form the saving balance.
Transaction Amalgamate (Amg) is for moving the saving and checking balances of an account to another account checking balance and
resetting the saving and checking balances of the first account to zero.
Transaction WriteCheck is for  withdrawing from a given account's checking balance.

\smallskip
\noindent
{\bf TPC-C \cite{TPCC}:} Which has five transactions that manipulate nine tables: WAREHOUSE, DISTRICT, STOCK, ITEMS, CUSTOMERS, HISTORY, ORDER, NEWORDER, and ORDERLINE.
Transaction NewOrder is for placing a new order on a set of items.
Transaction Delivery is for delivering a withstanding order at certain warehouse.
Transaction Payment is for a given customer paying a withstanding amount of credit.
Transaction OrderStatus is for inspecting certain orders and the associated order lines.
Transaction StockLevel is for inspecting stocks at certain warehouse and the withstanding orders at this warehouse.

\end{document}